\documentclass[sn-mathphys]{sn-jnl}

\jyear{2021}%
\usepackage{threeparttable,makecell,natbib}
\usepackage{graphicx}
\usepackage{epstopdf}
\usepackage{makecell}
\theoremstyle{thmstyleone}%
\newtheorem{theorem}{Theorem}[section]

\newtheorem{lemma}{Lemma}[section]
\newtheorem{corollary}{Corollary}[section]
\theoremstyle{thmstyletwo}%
\newtheorem{example}{Example}%
\newtheorem{remark}{Remark}%

\newcommand{\red}[1]{{{\color{black}#1}}}

\theoremstyle{thmstylethree}%
\newtheorem{definition}{Definition}%

\begin{document}



\title[Article Title]{\textcolor{black}{New sets of Non-Orthogonal Spreading Sequences With Low Correlation and Low PAPR Using Extended Boolean Functions}}


\author[1]{\fnm{Kaiqiang} \sur{Liu}}\email{KQLiu@my.swjtu.edu.cn}
\equalcont{These authors contributed equally to this work.}

\author*[1]{\fnm{Zhengchun} \sur{Zhou}}\email{zzc@swjtu.edu.cn}
\equalcont{These authors contributed equally to this work.}

\author[1]{\fnm{Avik Ranjan} \sur{Adhikary}}\email{avik.adhikary@swjtu.edu.cn}
\equalcont{These authors contributed equally to this work.}

\author[1]{\fnm{Rong} \sur{Luo}}\email{luorong@home.swjtu.edu.cn}
\equalcont{These authors contributed equally to this work.}

\affil*[1]{\orgdiv{Department of Mathematics}, \orgname{Southwest Jiaotong University}, \orgaddress{\street{Xipu}, \city{Chengdu}, \postcode{610031}, \state{Sichuan}, \country{China}}}




\abstract{\textcolor{black}{Extended Boolean functions (EBFs)} are one of the most important tools in cryptography and spreading sequence design in communication systems. In this paper, we {\red{use}} \textcolor{black}{EBFs} to design \textcolor{black}{new sets} of spreading sequences for non-orthogonal multiple access (NOMA), which is an emerging technique capable of supporting massive machine-type communications (mMTC) in 5G and beyond. In this work, first $p$-ary complementary sequences are constructed using \textcolor{black}{EBFs} and then, these sequences are used to design \textcolor{black}{new sets} of non-orthogonal spreading sequence sets having very low coherence and peak to average power ratio (PAPR). The proposed spreading sequence sets are capable of supporting a large number of active devices simultaneously. In fact, for a $p$-ary spreading sequence set, we theoretically {\red{achieve}} an overloading factor of $2p$, where $p$ is an odd prime. Specifically, for $p=3$, we {\red{achieve}} an overloading factor of $6$, which {\red{cannot}} be achieved through the existing constructions till date.
	
%
%
}




\maketitle

\section{Introduction}\label{sec1}
\textcolor{black}{Extended Boolean functions (EBFs)} are $p$-ary functions over finite field, which play a very important role in cryptography and designing spreading sequences for communication systems \cite{Mesnager2016}. Recently, Boolean functions have been used to design spreading sequences for non-orthogonal multiple access (NOMA) \cite{Nam2021,Nam2020,LYB2022CL}. NOMA is a multiple access scheme, in which, each user utilizes the time and frequency resources at the same time, where they are categorised by their power levels. At the transmitter, NOMA uses superposition coding, \red{where each user is assigned a unique spreading sequence}, such that the successive interference cancellation (SIC) receiver can separate the users both in the uplink and in the downlink channels. Owing to its various advantages like spectral efficiency and low latency, NOMA is a preferred technique for massive machine-type communication (mMTC), which is a new service category of 5G that can support an extremely high connection density of IoT (Internet of Things) devices \cite{surveyDai2015,surveyDai2018}. In mMTC, at a specific time slot, only a small number of devices are active. This leads to the use of compressed sensing (CS) for joint channel estimation (CE) and multiuser detection (MUD) in NOMA \cite{BookEldar2012,Du2018}.

For a successful compressed sensing based detection, the coherence of the spreading sequences should be sufficiently low \cite{BookEldar2012}. {\red{Also, when}} the transmitted signals are spread over multiple subcarriers, the peak-to-average power ratio (PAPR) should be sufficiently low, to ensure that the transmitted signals of the active devices does not suffer from the PAPR problem \cite{BookLitsyn2007}. In this scenario, designing spreading sequences with low coherence as well as low PAPR is an important research topic for uplink grant-free NOMA. 


\textcolor{black}{In general, designing sequence sets with low coherence and low PAPR is a very challenging problem. In literature, complementary sequences are explored in detail for their low PAPR properties \cite{GDJ1999,Paterson2000}. Complementary sets are sets of sequences, in which autocorrelation of each of the sequences sum up to zero at each non-zero time-shifts. Generalized Boolean functions (GBFs) \cite{GDJ1999,Paterson2000,chen2016cs,chen2018cs,Shen2020cs} and other recusrsive constructions \cite{Avik2019cs,Avik2020cs} are extensively used to construct complementary sets with various parameters.} To satisfy the additional property of low coherrence along with low PAPR, various sequences such as pseudo-random noise, quasi-orthogonal,  and random Gaussian sequences have been analysed \cite{Yang2000,Zhang2022,Du2017,Liang2018}. Recently, Yu proposed a systematic framework based on GBFs to construct non-orthogonal spreading sequences with low PAPR for reliable multiuser detection in uplink grant-free access \cite{Nam2021}. In \cite{Nam2021}, exhaustive computer search were carried out to find a set of permutations, which can be used to ensure low coherence of the spreading sequence matrix. In another work of Yu \cite{Nam2020}, spreading sequence matrices with low coherence and low PAPR are proposed, by exploiting the graph structure associated with GBFs. Tian \textit{et al.} \cite{LYB2022CL} exploited the relationship between the coherence of the spreading matrices with the Walsh transform of the GBFs, to design spreading sequence matrices with low coherence. Note that, since all the constructions are based on GBFs, the spreading sequences are constructed over alphabet size $2^h$, where $h$ is an integer.

In literature, EBFs are used to design $p$-ary complementary sequences with low PAPR \cite{WZL2021}. Motivated by the works of \cite{Nam2021,Nam2020,LYB2022CL}, we use EBFs to design sequence sets with low coherence as well as low PAPR.
%
%
%
%
First, we represent the quadratic EBFs in the form of matrices and then, relate coherence of the spreading sequences, corresponding to the function, with the rank of those matrices. The resultant spreading sequence matrices are over alphabet of size $p^h$, where $p$ is an odd prime and $h$ is an integer. Also, the length of the sequences are of the form $p^m$ where $m$ is an integer. This largely extends the constructions reported in \cite{Nam2020} and improves flexibility of the parameters of the spreading sequence sets. The contributions of the paper can be listed as follows:
\begin{itemize}
	\item Using EBFs, three new sets of $p$-ary spreading sequence matrices are proposed.
	\item The overloading factor can be increased {\red{up to}} $2p$, in contrast to $3$ in \cite{Nam2020} and $4$ in \cite{LYB2022CL}. This significantly increases the capability to support more active devices in NOMA, when choosing large $p$.
\end{itemize}
The parameters of the spreading sequences matrices for uplink NOMA, proposed till date, through algebraic constructions are listed in Table \ref{tab1}.

The rest of the paper is organised as follows. In Section \ref{prelinimaries}, we have fixed some notations and revisited basic definitions. In Section \ref{section3}, we have recalled some recent results about designing spreading sequences using EBFs. In this section, we have also established the relationship between rank of the symplectic matrix related to the quadratic EBFs and the coherence of the spreading sequence matrix. In Section \ref{sec4}, we have proposed several frameworks to design infinite families of non-orthogonal spreading sequence matrices with optimum and near-optimum coherence and low PAPR. In this section, we have also shown that one of the proposed construction can be considered as an extended version of a previous construction reported in \cite{Nam2020}. Finally, we conclude the paper in Section \ref{sec5}.

\section{Preliminaries}\label{prelinimaries}
Before we begin, let us define some notations, which will be used throughout this paper.
\begin{itemize}
	\item \red{$\mathbf{0}_m$ and $\mathbf{1}_m$ denote $m$ number of consecutive $0$'s and $1$'s, respectively.}
	\item Let $p$ be an odd prime, $\mathbb{F}_{p}$ denotes the finite field with $p$ elements, {\red{$\mathbb{F}_{p}^{*}=\mathbb{F}_{p} \setminus\{0\}$}}.
	\item $\omega_{p}=e^{\sqrt{-1}2\pi/p}$ denotes the $p$-th root of unity.
	\item A matrix (or vector) is represented by a bold-face upper (or lower) case letter.
	\item $\langle \mathbf{a},\mathbf{b}\rangle$ denotes the inner-product of $\mathbf{a}$ and $\mathbf{b}$. 
	\item $a^*$ denotes the conjugate of $a$.
	\item $rank_p(\mathbf{A})$ denotes the rank of the matrix $\mathbf{A}$ over field $\mathbb{F}_p$.
\end{itemize}

\subsection{Peak-to-Average Power Ratio}
In a multicarrier communication, the transmitted signal of a device is spread onto multiple subcarriers using a spreading sequence. The spreading sequences with high peak-to-average power ratio (PAPR) will cause signal distortion due to the non-linearity of power amplifiers, which will further affects the performance of the transmission system.
Suppose that $\mathbf{s} = (s_{0},\cdots,s_{M-1})^{T}\in\mathbb{C}^{M}$ is a spreading sequence which is used to transmitted via $M$ subcarriers. The peak-to-average power ratio (PAPR) of $\mathbf{s}$'s multicarrier signal
is defined by \cite{BookLitsyn2007} 
\begin{equation}
PAPR (\mathbf{s}) = \frac{1}{\sum_{i=0}^{M-1}\mid s_{i}\mid ^2} \cdot \max_{t\in [0,1)} 	 \mid \sum_{k=0}^{M-1} s_{k} e^{\sqrt{-1}2\pi k t}  \mid 
\end{equation}
Moreover, the PAPR of a sequence set $\mathbf{\Phi}$ is defined as:
\begin{equation}
	{\rm PAPR}(\mathbf{\Phi}) = \max_{\mathbf{s}\in \mathbf{\Phi}} {\rm PAPR}(\mathbf{s}).
\end{equation}

In general, for a given sequence, it is not easy to determine its PAPR. Fortunately, for complementary sequences, there {\red{exists}} an upper bound on their PAPR.

\subsection{Complementary sequences}
\begin{definition}
	Let $\mathbf{a}=\{a_j\}_{j=0}^{M-1}$, $\mathbf{b}=\{b_j\}_{j=0}^{M-1}$ be two complex-valued sequence of length $M$. The aperiodic correlation function $R_{\mathbf{a},\mathbf{b}}(\tau)$ of $\mathbf{a}$ and $\mathbf{b}$ at time delay $\tau$ is defined by
	\begin{equation}\label{eq_sum_of_autocorrelation_of_CS}
		R_{\mathbf{a},\mathbf{b}}(\tau)=\begin{cases}
			\sum_{k=0}^{M-1-\tau} a_{k}b_{k+\tau}^{*}, &0\leq \tau\leq M-1,\\
			\sum_{k=0}^{M-1+\tau} a_{k-\tau}b_{k}^{*}, &1-M\leq \tau\leq -1,\\
			0,&  \mid \tau \mid  \geq M.
		\end{cases}
	\end{equation}
	When $\mathbf{a}=\mathbf{b}$,  then it is called the aperiodic autocorrelation function, and is denoted by $R_{\mathbf{a}}(\tau)$. 
\end{definition}

\begin{definition}
	Let $\mathbf{S}$ be a set of $N$ sequences, each of length $M$, then $\mathbf{S}$ is called a complementary set (CS) with parameters $(N,M)$ if 
	\begin{equation}
		\sum_{i=0}^{N-1} R_{\mathbf{a}_i}(\tau) = 0, ~~1\leq \tau\leq M-1,
	\end{equation}
where $\mathbf{a}_i$ denotes the $i$-th sequence of $\mathbf{S}$.
\end{definition}

The following lemma, \textcolor{black}{proposed in \cite{Paterson2000}}, gives the upper bound about the PAPR of complementary sequences.
\begin{lemma}{{\rm  {\red{(\cite{Paterson2000}}})}\label{lemma_upperboundPAPR_of_CSS}}
	Let $\mathbf{S}$ be an $(N,M)$-CS. Then the PAPR of any sequence in $\mathbf{S}$ is upper bounded by $N$.
\end{lemma}

\subsection{Coherence of sequence set}
Let $\mathbf{a},\mathbf{b}\in\mathbb{C}^{M}$, then $\mathbf{a}$ {\red{and}} $\mathbf{b}$ {\red{are}} said orthogonal if $\langle \mathbf{a},\mathbf{b}\rangle=0$.
Let $\mathbf{S}=[\mathbf{s}_{1},\ldots,\mathbf{s}_{N}]$ be a sequence set with $N$ sequences, each of length $M$. Then $\mathbf{S}$ is called an orthogonal sequence set if the sequences in $\mathbf{S}$ are pairwise orthogonal. If $\mathbf{S}$ is non-orthogonal, the coherence of $\mathbf{S}$ is defined as follows:
\begin{equation}\label{eq_coherence}
	\mu(\mathbf{S})=\max\limits_{1\leq i \neq j\leq N}\frac{\mid \langle \mathbf{s}_{i},\mathbf{s}_{j}\rangle \mid}{\| \mathbf{s}_{i}\|_{2} \|\mathbf{s}_{j}\|_{2}},
\end{equation}
where $\|\mathbf{s}_i\|_{2} = \sqrt{(\sum_{k=1}^{M} \mid s_{ik} \mid^2)}$ denotes the $\ell_{2}$-norm of $\mathbf{s}_i$, and $s_{ik}$ is the $k$-th element of the $i$-th sequence of $\mathbf{S}$.

\begin{definition}
	Suppose that $\mathbf{S}$ is a non-orthogonal sequence set with $N$ sequences, each of length $M$, then the overloading factor of $\mathbf{S}$ is defined as $\lceil N/M \rceil$ \cite{Nam2021}.
\end{definition}


\subsection{\textcolor{black}{Extended Boolean functions (EBFs)}}
Extended Boolean functions (EBFs) are functions on $m$ variables, from $\mathbb{F}_{p}^{m}$ to $\mathbb{F}_{p}$, represented as
\begin{equation}
	f(\mathbf{x})=\sum_{\mathbf{b}\in\mathbb{F}_{p}^{m}} \lambda_{\mathbf{b}} (\prod_{i=1}^{m} x_{i}^{b_{i}}),
\end{equation}
where $\lambda_{\mathbf{b}}\in\mathbb{F}_{p}$, $\mathbf{b}=(b_{1},\ldots,b_{m})\in\mathbb{F}_{p}^{m}$, and $\mathbf{x}=(x_{1},x_{2},\ldots,x_{m})$ be the $p$-ary representation of the integer $x=\sum_{k=1}^{m} x_{k}p^{k-1}$  \cite{SBS2021}. 

Given $f(\mathbf{x})$, the corresponding complex-valued sequence $\mathbf{s}_{f}$ of length $p^m$ is generated as follows:
\begin{equation}
	\mathbf{s}_{f}=(\omega_{p}^{f_{0}},\omega_{p}^{f_{1}},\ldots,\omega_{p}^{f_{p^m-1}}),
\end{equation}
where $f_{i}=f(i_{1},i_{2},\ldots,i_{m})$, and $i=\sum_{k=1}^{m} i_{k}p^{k-1}$.

\begin{example}\label{example_p-ary_functions}
	Let $p=3,m=2$, $\mathbf{b}=\{(b_1,b_2)\mid b_i\in \mathbb{F}_3, i=1,2\}$, $\lambda_{(0,0)}=1,\lambda_{(1,1)}=2,\lambda_{(2,0)}=1$ and $\lambda_{\mathbf{b}}=0$ for any other $\mathbf{b}\in \mathbb{F}_{3}^{2}$. Then
	\begin{equation}
		\begin{split}
			f(\mathbf{x})&=\sum_{\mathbf{b}\in\mathbb{F}_{3}^{2}} \lambda_{\mathbf{b}} (\prod_{i=1}^{m} x_{i}^{b_{i}})\\
			&=1+2x_1x_2+x_1^2,
		\end{split}
	\end{equation}
	and the associated sequence $\mathbf{s}_f=(1,2,2,1,1,0,1,0,1)$, where the entries are power of $\omega_{3}$.
\end{example}

\section{Spreading Sequences using EBFs}\label{section3}
To design a non-orthogonal spreading sequence set with low PAPR and low coherence, we will use sequence sets having low correlation and theoretically bounded PAPR. In this paper, we have considered the complementary sequences generated by EBFs, as the sequences with low PAPR.

\subsection{Complementary sequences using EBFs}

\begin{lemma}\label{p-CSS} {\rm \cite[Theorem 1]{SBS2021}}
	For $0\leq n\leq p-1$, let $f_{n}(\mathbf{x})$ {\red{denotes}} the EBF of the following form:
	\begin{equation}\label{eq9}
		f_n(\mathbf{x}) = \sum_{i=1}^{m-1} a_{i} x_{\pi(i)} x_{\pi(i+1)} +\sum_{t=1}^{p-1} \sum_{k=1}^{m}c_{t,k} x_{k}^{t}+nx_{\pi(1)},
	\end{equation}
	where $\pi$ is a permutation over $\{1,\ldots,m\}$, $a_{i}\in\mathbb{F}_{p}^{*},c_{t,k}\in\mathbb{F}_{p}$. 
	Then the sequence set $\mathbf{S}=[\mathbf{s}_{f_{0}},\ldots,\mathbf{s}_{f_{p-1}}]$ forms a $(p,p^m)$-CS.
\end{lemma}

Suppose $f(\mathbf{x})$ is a $p$-ary function with the form in Lemma \ref{p-CSS}, then by Lemma \ref{lemma_upperboundPAPR_of_CSS},  ${\rm PAPR}(\mathbf{s}_{f})\leq p$. 

\subsection{Non-orthogonal Sequences from $p$-ary complementary sequences}\label{SS_from_p-CSS}

Since the degree of the EBFs in Lemma \ref{p-CSS} is at most $p-1$,  it is very challenging to calculate the coherence for a large $p$. Therefore, in this paper, we study the sequences generated by quadratic EBFs, whose coherence can be determined by the rank information of some sympletic matrices. 

Regardless of the constant term, the EBFs in Lemma \ref{p-CSS} can be represented as 
{\red{\begin{equation}\label{eq_quadratic_functions}
	f(\mathbf{x})=\sum_{i=1}^{m-1} a_{i} x_{\pi(i)} x_{\pi(i+1)} + \sum_{k=1}^{m} d_{k} x_{k}^{2}+ \sum_{k=1}^{m} c_{k}x_{k}=\mathbf{x}\mathbf{A}\mathbf{x}^{T}+\mathbf{c}\mathbf{x}^{T},
\end{equation}}}
where $\mathbf{c}=(c_{1},\ldots,c_{m})$, $\mathbf{a}=(a_{1},\ldots,a_{m-1})$, $\mathbf{d}=(d_{1},\ldots,d_{m})$ and $\mathbf{A}$ is a matrix determined by $\pi,~\mathbf{a},~\mathbf{d}$. Without loss of generality, let us use $\psi(\pi,\mathbf{a},\mathbf{d})$ to represent $\mathbf{A}$, where
\begin{equation}\label{eq_matrix}
	A(i,j) = \begin{cases}
		d_{i},& \text{if~~} i=j;\\
		a_{k},& \text{if~~} \pi(k)=i,\pi(k+1)=j;\\
		0,& \text{otherwise.}
	\end{cases}
\end{equation}	

Let $M_{m}(\mathbb{F}_{p})$ {\red{denote}} the set of all {\red{square}} matrices of {\red{size $m\times m$(or order $m$)}} over $\mathbb{F}_p$, and $\mathcal{A}_{p} \subset M_{m}(\mathbb{F}_{p})$ {\red{denote}} the set of all matrices which satisfies (\ref{eq_matrix}). Considering (\ref{eq9}) for a EBF of degree $2$, any quadratic EBF $f_{\mathbf{A}}^{(c)}(\mathbf{x}) = \mathbf{x}\mathbf{A}\mathbf{x}^{T}+\mathbf{c}\mathbf{x}^{T}$ with $\mathbf{A}\in\mathcal{A}_p$ and $c=\sum_{k=1}^{m} c_{k}p^{k-1}$ will generate a complex-valued sequence $\mathbf{s}_{f}$ with ${\rm PAPR}(\mathbf{s}_{f})\leq p$. For simplicity, we denote $\mathbf{s}_{\mathbf{A}}^{(c)}$ as the corresponding sequence of $f_{\mathbf{A}}^{(c)}$.

{\red{\begin{example}
			Let $p=3,m=4$, $\pi=[1,2,4,3]$ denote the permutation defined in $\{1,2,3,4\}$, $\mathbf{a}=(2,1,1),\mathbf{d}=(2,0,2,1)\in \mathbb{F}_{p}^{m}$. Let $f$ be the EBF defined as follows:
			\begin{equation}
				\begin{split}
					f(\mathbf{x}) &= \sum_{i=1}^{m-1} a_{i}x_{\pi(i)}x_{\pi(i+1)}+\sum_{k=1}^{m} d_{k}x_{k}^{2}\\
					&= 2x_{1}x_{2}+x_{2}x_{4}+x_{4}x_{3}+2x_{1}^{2}+x_{3}^{2}+x_{4}^{2}
				\end{split}
			\end{equation} 
			Then by (\ref{eq_quadratic_functions}), we have
			\begin{equation}
				f(\mathbf{x})=\mathbf{x}\mathbf{A}\mathbf{x}^{T},
			\end{equation}
			where 
			\begin{equation}
				\mathbf{A}=\psi(\pi,\mathbf{a},\mathbf{d})=\left(\begin{array}{cccc}
					2 & 2 & 0 & 0 \\
					0 & 0 & 0 & 1 \\
					0 & 0 & 1 & 0 \\
					0 & 0 & 1 & 1
				\end{array}
				\right).
			\end{equation}
\end{example}}}


\begin{remark}
	For $p=2$, since $x_{k}^{2}=x_{k}$ in $\mathbb{F}_{2}$, the quadratic binary functions in (\ref{eq_quadratic_functions}) can be expressed as $f(\mathbf{x})=\sum_{i=1}^{m-1} a_{i} x_{\pi(i)} x_{\pi(i+1)} +  \sum_{k=1}^{m} c_{k}x_{k}$. Then the corresponding quadratic matrix $\mathbf{A}$ is represented as $\psi(\pi,\mathbf{1}_{m-1},\mathbf{0}_{m})$.
\end{remark}

Let $M=p^m$, and $0\leq c <M$. Given $\mathbf{A}\in\mathcal{A}_{p}$, the sequence set generated by $\mathbf{A}$ is defined as follows:
\begin{equation}
	\mathbf{\Phi}_{\mathbf{A}}=[\mathbf{s}_{\mathbf{A}}^{(0)},\mathbf{s}_{\mathbf{A}}^{(1)},\ldots,\mathbf{s}_{\mathbf{A}}^{(M-1)}],
\end{equation}  
where $\mathbf{s}_{\mathbf{A}}^{(i)}$ is the $i$-th column of $\mathbf{\Phi}_{\mathbf{A}}$.
\begin{lemma}
	$\mathbf{\Phi}_{\mathbf{A}}$ is an orthogonal sequence set.
\end{lemma}
\begin{proof}
	Suppose that $\mathbf{s}_{\mathbf{A}}^{(a)}$ and $\mathbf{s}_{\mathbf{A}}^{(b)}$ are different sequences in $\mathbf{\Phi}_{\mathbf{A}}$ with $0\leq a\not=b \leq p^m-1$. Let $f_{1}(\mathbf{x})=\mathbf{x}^{T}\mathbf{A}\mathbf{x}+\mathbf{a}^{T}\mathbf{x}$ and $f_{2}(\mathbf{x})=\mathbf{x}^{T}\mathbf{A}\mathbf{x}+\mathbf{b}^{T}\mathbf{x}$ be the two EBFs associated to $\mathbf{s}_{\mathbf{A}}^{(a)}$ and $\mathbf{s}_{\mathbf{A}}^{(b)}$, where $\mathbf{a}=(a_{1},\ldots,a_{m}),\mathbf{b}=(b_{1},\ldots,b_{m})$ denotes the $p$-ary representation of $a$ and $b$. Then, we have
	\begin{equation}
		\begin{split}
			\langle \mathbf{s}_{\mathbf{A}}^{(a)},\mathbf{s}_{\mathbf{A}}^{(b)}\rangle&=\sum_{\mathbf{x}\in\mathbb{F}_{p}^{m}} \omega_{p}^{f_{1}(\mathbf{x})-f_{2}(\mathbf{x})}\\
			&=\sum_{\mathbf{x}\in\mathbb{F}_{p}^{m}} \omega_{p}^{(\mathbf{a}-\mathbf{b})^{T}\mathbf{x}}\\
			&=\sum_{\mathbf{x}\in\mathbb{F}_{p}^{m}} \omega_{p}^{\mathbf{c}^{T}\mathbf{x}}\\
			&=0 ,
		\end{split}
	\end{equation}
	since $\mathbf{c}=\mathbf{a}-\mathbf{b}\not=\mathbf{0}\in\mathbb{F}_{p}^{m}$. This completes the proof.
\end{proof}
 To obtain a large non-orthogonal sequence set, we expand the sequence set by exploiting more quadratic matrices from $\mathcal{A}_p$. Let $\mathcal{M}=\{\mathbf{A}_{1},\mathbf{A}_{2},\ldots,\mathbf{A}_{L}\}\subset \mathcal{A}_{p}$, the non-orthogonal sequence set generated by $\mathcal{M}$ is defined by
\begin{equation}\label{phimatrix}
	\mathbf{\Phi} = [\mathbf{\Phi}_{\mathbf{A}_{1}},\mathbf{\Phi}_{\mathbf{A}_{2}},\ldots, \mathbf{\Phi}_{\mathbf{A}_{L}}].
\end{equation}

From (\ref{eq_coherence}), the coherence of $\mathbf{\Phi}$ is given by
\begin{equation}\label{eq13av}
	\mu(\mathbf{\Phi}) = \max\limits_{\substack{1\leq i,j\leq L\\
			0\leq c_{1},c_{2}\leq M-1 }}\frac{\mid\langle \mathbf{s}_{\mathbf{A}_{i}}^{(c_{1})}, \mathbf{s}_{\mathbf{A}_{j}}^{(c_{2})}\rangle\mid}{M},
\end{equation}
where $c_{1}\not=c_{2}$ if $i=j$. 

Next, we will show that the coherence of $\mathbf{\Phi}$ depends on the rank information of the symplectic matrix corresponding to a pair of matrices in $\mathcal{M}$.

\begin{lemma}\label{lemma_correlation}
	Let  $p$ be an odd prime, $M=p^m$,  $\mathcal{M}=\{\mathbf{A}_{1},\mathbf{A}_{2}\}$ and $\mathbf{\Phi}$ is generated by $\mathcal{M}$, as defined in (\ref{phimatrix}). Define a
	symplectic matrix $\mathbf{Q}_{1,2} = (\mathbf{A}_{1}-\mathbf{A}_{2})+(\mathbf{A}_{1}-\mathbf{A}_{2})^{T}$. Then, if $rank_{p}(\mathbf{Q}_{1,2})=r$, 
	
	\begin{equation}
		\mu(\mathbf{\Phi})=\frac{1}{\sqrt{p^r}}.
	\end{equation}
\end{lemma}

\begin{proof}
	Let $f(\mathbf{x})=\mathbf{x}^{T} \mathbf{A}_1 \mathbf{x} + \mathbf{a_1}^{T} \mathbf{x}$, and $g(\mathbf{x})=\mathbf{x}^{T} \mathbf{A}_2  \mathbf{x} + \mathbf{a_2}^{T} \mathbf{x}$. Also let $\mathbf{s}_f$ and $\mathbf{s}_g$ be the corresponding spreading sequences of the EBFs $f$ and $g$, respectively. Then, we have
		\begin{equation}\label{eq1av}
		\begin{split}
			\mid \langle \mathbf{s}_{f}, \mathbf{s}_{g}\rangle \mid &=\mid \sum_{\mathbf{x}\in \mathbb{F}_{p}^{m}} \omega^{f(\mathbf{x})-g(\mathbf{x})} \mid \\
			&= \mid \sum_{\mathbf{x}\in \mathbb{F}_{p}^{m}} \omega^{\mathbf{x}^{T} (\mathbf{A} _1- \mathbf{A}_2) \mathbf{x} + (\mathbf{a_1}^{T}-\mathbf{a_2}^{T}) \mathbf{x}} \mid \\
			&= \mid \sum_{\mathbf{x}\in \mathbb{F}_{p}^{m}} \omega^{\mathbf{x}^{T} \mathbf{C} \mathbf{x} + \mathbf{c}^{T} \mathbf{x}} \mid (\mathbf{C}=\mathbf{A}_1-\mathbf{A}_2,\mathbf{c}=\mathbf{a_1}-\mathbf{a_2}) \\
			&= \left( \left(\sum_{\mathbf{x}\in \mathbb{F}_{p}^{m}} \omega^{\mathbf{x}^{T} \mathbf{C} \mathbf{x} + \mathbf{c}^{T} \mathbf{x}} \right) \left(\sum_{\mathbf{y}\in \mathbb{F}_{p}^{m}} \omega^{-(\mathbf{y}^{T} \mathbf{C} \mathbf{y} + \mathbf{c}^{T} \mathbf{y})} \right)\right)^{\frac{1}{2}} \\
			&= \left(\sum_{\mathbf{z}\in \mathbb{F}_{p}^{m}} \sum_{\mathbf{y}\in \mathbb{F}_{p}^{m}} \omega^{\mathbf{z}^{T} \mathbf{C} \mathbf{z} + \mathbf{c}^{T} \mathbf{z} } \cdot \omega^{\mathbf{y}^{T} \mathbf{C} \mathbf{z} + \mathbf{z}^{T} \mathbf{C} \mathbf{y}}\right)^{\frac{1}{2}} (\text{let~~} \mathbf{x}=\mathbf{y}+\mathbf{z}) \\
			&= \left(\sum_{\mathbf{z}\in \mathbb{F}_{p}^{m}} \omega^{\mathbf{z}^{T} \mathbf{C} \mathbf{z} + \mathbf{c}^{T} \mathbf{z} } \sum_{\mathbf{y}\in \mathbb{F}_{p}^{m}} \omega^{\mathbf{y}^{T} (\mathbf{C}+ \mathbf{C}^{T}) \mathbf{z}}\right)^{\frac{1}{2}}.
		\end{split}
	\end{equation}
We have,
\begin{equation}\label{eq2}
	\sum_{\mathbf{y}\in \mathbb{F}_{p}^{m}} \omega^{\mathbf{y}^{T} (\mathbf{C}+ \mathbf{C}^{T}) \mathbf{z}} =\begin{cases}
		p^m,& \text{if~~} (\mathbf{C}+\mathbf{C}^{T})\mathbf{z}=\mathbf{0};\\
		0,& \text{other cases}.
	\end{cases}
\end{equation}
Let $\Omega = \{\mathbf{z}\mid (\mathbf{C}+\mathbf{C}^{T})\mathbf{z}=\mathbf{0}, \mathbf{z} \in \mathbb{F}_{p}^{m}\}$. Then $\Omega$ is a subspace of $\mathbb{F}_{p}^{m}$. Let $s$ {\red{denote}} the dimension of $\Omega$, hence, the rank $r$ of $\mathbf{C}+\mathbf{C}^{T}$ is $m-s$. Then (\ref{eq1av}) can be written as follows:

\begin{equation}\label{eq3}
	\begin{aligned}
		\mid \langle \mathbf{s}_{f}, \mathbf{s}_{g}\rangle \mid &= \left(\sum_{\mathbf{z}\in \mathbb{F}_{p}^{m}} \omega^{\mathbf{z}^{T} \mathbf{C} \mathbf{z} + \mathbf{c}^{T} \mathbf{z} } \sum_{\mathbf{y}\in \mathbb{F}_{p}^{m}} \omega^{\mathbf{y}^{T} (\mathbf{C}+ \mathbf{C}^{T}) \mathbf{z}}\right)^{\frac{1}{2}}\\
		&= \left(p^m \sum_{\mathbf{z}\in \Omega} \omega^{\mathbf{z}^{T} \mathbf{C} \mathbf{z} + \mathbf{c}^{T} \mathbf{z} } \right)^{\frac{1}{2}}\\
		&= \left(p^m \sum_{\mathbf{z}\in \Omega} \omega^{ \mathbf{c}^{T} \mathbf{z} } \right)^{\frac{1}{2}},
	\end{aligned}
\end{equation}
since, $\mathbf{z}^{T} \mathbf{C} \mathbf{z} = (-\mathbf{z}^{T} \mathbf{C}^{T} \mathbf{z})^{T}=-\mathbf{z}^{T} \mathbf{C} \mathbf{z}$ for $\mathbf{z}\in \Omega$,
which means that $\mathbf{z}^{T} \mathbf{C} \mathbf{z}=0$.
Hence, 
\begin{equation}
	\mid \langle \mathbf{s}_{f}, \mathbf{s}_{g}\rangle \mid \leq p^{\frac{m+s}{2}}.
\end{equation}
Since, $s=m-r$, using (\ref{eq13av}), the proof is completed.
\end{proof}

\begin{remark}
	Note that, when $r=m$, the $p$-ary function $\mathbf{x}(\mathbf{C}+\mathbf{C}^{T}) \mathbf{x}^{T}$ corresponding to the matrix $\mathbf{C}=\mathbf{A}_{1}-\mathbf{A}_{2}$ is a generalized bent function \cite{KUMAR1985}.
\end{remark}

The following theorem establishes the relation between the coherence of $\mathbf{\Phi}$ and the rank of the matrices in $\mathcal{M}$, which is also discussed in \cite[Theorem 1]{Nam2021} for $p=2$.
\begin{theorem}{\rm  (\cite{Nam2021})\label{coherence}}
	Suppose that $\mathbf{\Phi}$ is the sequence set generated by $\mathcal{M}=\{\mathbf{A}_{1},\ldots,\mathbf{A}_{L}\}$, as defined in (\ref{phimatrix}). The coherence of $\mathbf{\Phi}$ is given by 
	\begin{equation}
		\mu(\mathbf{\Phi}) = \frac{1}{\sqrt{p^{r_{min}}}},
	\end{equation}
	where \begin{equation}
		r_{min} = \min_{1\leq i\not=j \leq L} rank_{p}(\mathbf{Q}_{i,j}).
	\end{equation}
\end{theorem}
\begin{proof}
	The proof can be directly derived from Lemma \ref{lemma_correlation}.
\end{proof}

Hence, from the above theorem, we observe that, to design a spreading sequence matrix $\mathbf{\Phi}$ with very small value of $\mu(\mathbf{\Phi})$, we need to get $r_{min}$ as large as possible. In this paper, the coherence of $\mathbf{\Phi}$ is optimum, when the rank of the corresponding symplectic matrix achieves its maximum value, i.e., $r_{min}=m$ which implies $\mu(\mathbf{\Phi})=\sqrt{\frac{1}{M}}$. When $r_{\min}=m-1$, then $\mu(\mathbf{\Phi})=\frac{1}{\sqrt{p^{m-1}}}$ and we call it near-optimum.


\section{Proposed spreading sequence sets with low coherence and low PAPR\label{sec4}}
In this section, we will propose several new non-orthogonal sequence set $\mathbf{\Phi}$ with large $L$ for $p\geq 3$.
\subsection{Spreading sequences of length $M=p^m$ for $p\geq 3$ \label{sec4_1}}
Recall that our objective is to find a class of matrices $\mathcal{M}\subset\mathcal{A}_{p}$ with $r_{min}\geq m-1$. Note that, if we consider a special case that $\forall\mathbf{A}_{i},\mathbf{A}_{j}\in \mathcal{M}$, $\mathbf{A}_{i}-\mathbf{A}_{j}$ is a diagonal matrix with full rank, then we have $r_{min}=m$. Based on this special case, we give the following construction.
\begin{theorem}\label{th_coherence_Lp}
	Let $\pi$ be a permutation over $\{1,\cdots,m\}$ and $\mathbf{a} = (a_{1},\cdots,a_{m-1})$ with $a_{i}\not=0~(1\leq i \leq m-1)$. Let   $\mathbf{d}_{1},\mathbf{d}_{2},\cdots,\mathbf{d}_{p}$ be $p$ vectors defined over $\mathbb{F}_{p}^{m}$, which satisfy the condition that $\forall k \in \{1,\cdots,m\},  d_{i,k}\not=d_{j,k} (1\leq i<j\leq p)$. Let $\mathcal{M}=\{\mathbf{A}_{1},\mathbf{A}_{2},\cdots,\mathbf{A}_{p}\}$ where
	\begin{equation}
			\begin{aligned}
				\mathbf{A}_{1} &= \psi(\pi,\mathbf{a},\mathbf{d}_{1}), \\
				\mathbf{A}_{2} &= \psi(\pi,\mathbf{a},\mathbf{d}_{2}), \\
				&~~\vdots \\
				\mathbf{A}_{p} &= \psi(\pi,\mathbf{a},\mathbf{d}_{p}). \\
			\end{aligned}		
	\end{equation}
 Let $\mathbf{\Phi}$ be the sequence set generated by $\mathcal{M}$, as defined in (\ref{phimatrix}). Then {\red{ $\mathbf{\Phi}$ is a sequence set with $p^{m+1}$ sequences, each of length $p^m$, and $\mu(\mathbf{\Phi})=\sqrt{\frac{1}{M}}$, $PAPR(\mathbf{\Phi})\leq p$ }}.	
\end{theorem}

\begin{proof}
 Suppose that $ i,j\in \{1,\ldots,p\}$ with $i\not=j$, and recall from Lemma \ref{lemma_correlation} that $\mathbf{Q}_{i,j} = (\mathbf{A}_{i}-\mathbf{A}_{j})+(\mathbf{A}_{i}-\mathbf{A}_{j})^{T} \in \mathbb{F}_{p}^{m\times m}$. Then, we have
			\begin{equation}
				\begin{split}
					\det(\mathbf{Q}_{i,j}) &= \begin{vmatrix}
						2(d_{i,1}-d_{j,1})	&  &  & \\
						& 2(d_{i,2}-d_{j,2}) &  &   \\
						&  & \ddots &  \\
						&  &  & 2(d_{i,m}-d_{j,m})
					\end{vmatrix}
					\\
					& = 2^m \prod_{k=1}^{m} (d_{i,k}-d_{j,k}).
				\end{split}
			\end{equation}
			Since $d_{i,k}\not=d_{j,k} ~(1\leq i<j\leq p)$, we have $\det(\mathbf{Q}_{i,j})\not=0$, and hence $rank_{p}(\mathbf{Q}_{i,j})=m$. 
			According to Theorem \ref{coherence}, the coherence of $\mathbf{\Phi}$ is given by 	
			\begin{equation}
				\mu(\mathbf{\Phi}) = \frac{1}{\sqrt{p^{r_{min}}}} = \frac{1}{\sqrt{p^{m}}} = \sqrt{\frac{1}{M}},
			\end{equation} 
			since $r_{min} = \min_{1\leq i\not=j \leq L} rank_{p}(\mathbf{Q}_{i,j})$.
%
\end{proof}

{\red{\begin{example}\label{example_4_1}
			Let $p=5,m=3$, $\pi=[3,1,2]$, $\mathbf{a}=(2,2)$, and $\mathbf{d}_{1}=(0,3,4),\mathbf{d}_{2}=(1,0,1),\mathbf{d}_{3}=(2,1,2),\mathbf{d}_{4}=(3,2,0),\mathbf{d}_{5}=(4,4,3)$. According to Theorem \ref{th_coherence_Lp} and (\ref{eq_matrix}), we get the following $5$ matrices:
			\begin{equation*}
				\begin{split}
					&	\mathbf{A}_1=\left(\begin{array}{ccc}
						0 & 2 & 0\\
						0 & 3 & 0\\
						2 & 0 & 4
					\end{array}\right),
					\mathbf{A}_2=\left(\begin{array}{ccc}
						1 & 2 & 0\\
						0 & 0 & 0\\
						2 & 0 & 1
					\end{array}\right),
					\mathbf{A}_3=\left(\begin{array}{ccc}
						2 & 2 & 0\\
						0 & 1 & 0\\
						2 & 0 & 2
					\end{array}\right),\\ &\mathbf{A}_4=\left(\begin{array}{ccc}
						3 & 2 & 0\\
						0 & 2 & 0\\
						2 & 0 & 0
					\end{array}\right),
					\mathbf{A}_5=\left(\begin{array}{ccc}
						4 & 2 & 0\\
						0 & 4 & 0\\
						2 & 0 & 3
					\end{array}\right).
				\end{split}
			\end{equation*}
			Here, $rank_{5}(\mathbf{Q}_{i,j})=3$ for any $i\not=j$. 	
			Let $\mathbf{\Phi}=\frac{1}{5^{3}}[\mathbf{\Phi}_{\mathbf{A}_1},\mathbf{\Phi}_{\mathbf{A}_2},\mathbf{\Phi}_{\mathbf{A}_3},\mathbf{\Phi}_{\mathbf{A}_4},\mathbf{\Phi}_{\mathbf{A}_5}]$ be the $5$-ary spreading matrix of size $5^{3}\times 5^{4}$. Then, ${\rm PAPR}({\mathbf{\Phi}})=3.5223<5$ and $\mu(\mathbf{\Phi})=\sqrt{\frac{1}{5^{3}}}\approx0.0894$, which is optimum.
\end{example}}}


{\red{\begin{remark}\label{number_of_phi}
For a given permutation $\pi$ and $\mathbf{a}\in \mathbf{F}_{p}^{m}$, let $\mathcal{N}_{\mathbf{\Phi}}$ denotes the number of $\mathbf{\Phi}$. Then $\mathcal{N}_{\mathbf{\Phi}}$ equals to all the different choice of $(\mathbf{d}_{1},\ldots,\mathbf{d}_{p})$. Hence, we have 
\begin{equation}\label{eq_number_of_phi_1}
	\mathcal{N}_{\mathbf{\Phi}}=(p!)^{m-1}.
\end{equation}
\end{remark}
}}

Before we proceed, we need the following lemma.

\begin{lemma}\label{lemma_rank_calA}
	Let $\mathcal{A}_{p}$ be the set of matrices defined in Section \ref{section3}. Then $\forall \mathbf{A}\in\mathcal{A}_{p}$, 
	\begin{equation}
		rank_{p}(\mathbf{A}+\mathbf{A}^{T})\geq m-1.
	\end{equation} 
\end{lemma}
\begin{proof}
	Suppose that $\mathbf{A}=\psi(\pi,\mathbf{a},\mathbf{d})\in \mathcal{A}_{p}$. Let $f_{\mathbf{A}}(\mathbf{x})=\mathbf{x}\mathbf{A}\mathbf{x}^{T}$. For any $\mathbf{x}\in\mathbb{F}_p^{m}$, let $\mathbf{P}$ {\red{denotes}} the permutation matrix defined by $\pi$, i.e. for any $\mathbf{x}=(x_{1},\cdots,x_{m})$, $\mathbf{x}\mathbf{P}=(x_{\pi(1)},\cdots,x_{\pi(m)})$. Suppose that $\mathbf{y}=\mathbf{x}\mathbf{P}$ and $g_{\mathbf{A}}(\mathbf{x}) = \mathbf{x}(\mathbf{A}+\mathbf{A}^{T})\mathbf{x}^{T}$, then,
	\begin{equation}
		\begin{aligned}
			g_{\mathbf{A}}(\mathbf{x}) &= \mathbf{x}(\mathbf{A}+\mathbf{A}^{T})\mathbf{x}^{T}\\
			&= \sum_{i=1}^{m-1} 2 a_{i}x_{\pi(i)}x_{\pi(i+1)} + \sum_{j=1}^{m} 2 d_{j} x_{j}^{2}\\
			&= \sum_{i=1}^{m-1} 2 a_{i}y_{i}y_{i+1} + \sum_{j=1}^{m} 2 d_{\pi(j)} y_{j}^{2}\\
			&=\mathbf{y}(\mathbf{B}+\mathbf{B}^{T})\mathbf{y}^{T}\\
			&=\mathbf{x} \mathbf{P}(\mathbf{B}+\mathbf{B}^{T})\mathbf{P}^{T}\mathbf{x}^{T},
		\end{aligned}
	\end{equation}  
	where $\mathbf{B}=\psi(\pi^\prime,\mathbf{a},\mathbf{d}^{'})$ with $d_{i}^{'}=d_{\pi(i)}, \text{ and }\pi^\prime(i)=i \text{ for } 1\leq i \leq m$. 
	Since $\mathbf{P}$ is a permutation matrix with $\mathbf{P}\mathbf{P}^{T}=\mathbf{I}$, we have
	\begin{equation}
		rank_{p}(\mathbf{A}+\mathbf{A}^{T})=rank_{p}(\mathbf{P}(\mathbf{B}+\mathbf{B}^{T})\mathbf{P}^{T})=rank_{p}(\mathbf{B}+\mathbf{B}^{T}).
	\end{equation}
	Note that $\mathbf{B}+\mathbf{B}^{T}$ is a tridiagonal matrix, which can be written as
	\begin{equation}
		\left(\begin{array}{cccccc}
			2d_{1}^{'} & a_{1} & & & & \\
			a_{1} & 2d_{2}^{'} & a_{2} & & & \\
			& a_{2} & 2d_{3}^{'} & a_{3} & & \\
			& & \ddots & \ddots & \ddots & \\
			& & & a_{m-2} & 2d_{m-1}^{'} & a_{m-1} \\
			& & & & a_{m-1} & 2d_{m}^{'}
		\end{array}\right),
	\end{equation}
	by considering its top right $(m-1)\times (m-1)$ sub-matrix $\mathbf{B}^{'}$, we have 
	\begin{equation}
		\begin{split}
			\det(\mathbf{B}^{'}) &= \begin{vmatrix}
				a_{1} &  & & & & \\
				2d_{2}^{'} & a_{2} &  & & & \\
				a_{2} & 2d_{3}^{'} & a_{3} &  & & \\
				& \ddots & \ddots & \ddots &  & \\
				& & a_{m-3} & 2d_{m-2}^{'} & a_{m-2} &  \\
				& & & a_{m-2} & 2d_{m-1}^{'} & a_{m-1}
			\end{vmatrix}\\
			& = \prod_{k=1}^{m-1} a_{k} \not=0,
		\end{split}
	\end{equation}
	i.e. $rank_{p}(\mathbf{B}^{'})=m-1$.
	Then $rank_{p}(\mathbf{B}+\mathbf{B}^{T})\geq rank_{p}(\mathbf{B}^{'})=m-1$.
\end{proof}

Next, we propose two more constructions, generalizing the construction proposed in Theorem \ref{th_coherence_Lp}.

\begin{theorem}\label{construction2} 
	Let $\pi$ be a permutation over $\{1,\ldots,m\}$. Let $\mathbf{a}=(a_{1},\ldots,a_{m-1})$,$\mathbf{b}=(b_{1},\ldots,b_{m-1}) \in \mathbb{F}_{p}^{m-1}$, which satisfy the condition that $\forall i\in\{1,\cdots,m-1\},a_{i}-b_{i}\not=0$. Let the set of matrices $\mathcal{M}=\{\mathbf{A}_{1},\cdots,\mathbf{A}_{2p}\}$ be defined as follows:
	\begin{equation}
		\mathbf{A}_{i}=\begin{cases}
			\psi(\pi,\mathbf{a},\mathbf{d}_{i}), &1\leq i\leq p,\\
			\psi(\pi,\mathbf{b},\mathbf{d}_{i}), &p+1\leq i\leq 2p,
		\end{cases}
	\end{equation}
	where $\{\mathbf{d}_{1},\cdots,\mathbf{d}_{p}\},\{\mathbf{d}_{p+1},\cdots,\mathbf{d}_{2p}\}$ denote two sets of vectors each of which satisfy the conditions in Theorem \ref{th_coherence_Lp}. Let $\mathbf{\Phi}$ {\red{be the spreading sequence set}} generated by $\mathcal{M}$, as defined in (\ref{phimatrix}). Then {\red{$\mathbf{\Phi}$ is a sequence set with $2p^{m+1}$ sequences, each of length $p^m$, and}} ${\rm PAPR}(\mathbf{\Phi})\leq p$, $\mu(\mathbf{\Phi}) \leq \sqrt{\frac{p}{M}}$.
\end{theorem}

	\begin{proof}	
		From Lemma \ref{p-CSS}, we have ${\rm PAPR}(\mathbf{\Phi}_{\mathbf{A}_{i}})\leq p$ for any $1\leq i\leq 2p$, hence, eventually we have ${\rm PAPR}(\mathbf{\Phi})\leq p$.
		
		By Theorem \ref{coherence}, the coherence of $\mathbf{\Phi}$ depends on the minimum rank of $\mathbf{Q}_{i,j} = (\mathbf{A}_{i}-\mathbf{A}_{j})+(\mathbf{A}_{i}-\mathbf{A}_{j})^{T}$, where $1\leq i<j\leq 2p$. Recall from the proof of Theorem \ref{th_coherence_Lp}, when $i,j\in\{1,\cdots,p\}$ or $i,j\in\{p+1,\cdots,2p\}$, $\mathbf{Q}_{i,j}$ is a diagonal matrix with full rank, i.e. $rank_{p}(\mathbf{Q}_{i,j})=m$. Next, consider the case, when $i\in\{1,\cdots,p\},j\in\{p+1,\cdots,2p\}$, then we have
		\begin{equation}
			\begin{split}
				\mathbf{A}_{i}-\mathbf{A}_{j}&=\psi(\pi,\mathbf{a},\mathbf{d}_{i})-\psi(\pi,\mathbf{b},\mathbf{d}_{j})\\
				&=\psi(\pi,\mathbf{a}-\mathbf{b},\mathbf{d}_{i}-\mathbf{d}_{j}).
			\end{split}
		\end{equation}
		Since $a_{i}-b_{i}\not=0(1\leq i\leq m-1)$,  we have $\mathbf{A}_{i}-\mathbf{A}_{j}\in\mathcal{A}_{p}$. By Lemma \ref{lemma_rank_calA}, $rank_{p}(\mathbf{Q}_{i,j})\geq m-1$. Recalling $M=p^m$, using Theorem \ref{coherence}, we get $\mu(\mathbf{\Phi}) \leq \sqrt{\frac{p}{M}}$.
\end{proof}


\begin{theorem}\label{construction3}
	Let $\pi$ be a permutation over $\{1,\ldots,m\}$.
	Let $\pi^{\tau}$ {\red{denote}} the cyclic $\tau~(\tau>0)$ shift of $\pi$,  where $\pi^{\tau}(i) = \pi([i+\tau]_{m})$ with
	\begin{equation}
		[k]_{m} = \begin{cases}
			m, & \text{if $m\mid k$};\\
			v,& \text{if~} k=sm+v,v\not=0.
		\end{cases}
	\end{equation}
	Let $\mathbf{a}=(a_{1},\ldots,a_{m-1}),\mathbf{b}=(b_{1},\ldots,b_{m-1}) \in \mathbb{F}_{p}^{m-1}$, which satisfy the following two conditions:
	
	(1) $\forall i\in \{1,\cdots,m-1\}$, $a_{i}\not=0,b_{i}\not=0$;
	
	(2) there exist one $t\in\{1,\cdots,m-1\}\setminus\{m-\tau\}$, such that $\forall i\in\{1,\cdots,m-1\}\setminus\{m-\tau,t\}$,$b_{i}-a_{[i+\tau]_{m}}\not=0$ and $b_t-a_{[t+\tau]_{m}}=0$.
	Let $\mathcal{M}=\{\mathbf{A}_{1},\cdots,\mathbf{A}_{2p}\}$ with 
	\begin{equation}
		\mathbf{A}_{i}=\begin{cases}
			\psi(\pi,\mathbf{a},\mathbf{d}_{i}), &1\leq i\leq p,\\
			\psi({\red{\pi^{\tau}}},\mathbf{b},\mathbf{d}_{i}), &p+1\leq i\leq 2p,
		\end{cases}
	\end{equation}
	where $\{\mathbf{d}_{1},\cdots,\mathbf{d}_{p}\},\{\mathbf{d}_{p+1},\cdots,\mathbf{d}_{2p}\}$ denote two sets of vectors each of which satisfy the conditions in Theorem \ref{th_coherence_Lp}. Let $\mathbf{\Phi}$ {\red{be the spreading sequence set}} generated by $\mathcal{M}$, as defined in (\ref{phimatrix}). Then {\red{$\mathbf{\Phi}$ is a sequence set with $2p^{m+1}$ sequences, each of length $p^m$, and}} ${\rm PAPR}(\mathbf{\Phi})\leq p$, and $\mu(\mathbf{\Phi}) \leq \sqrt{\frac{p}{M}}$.
\end{theorem}

%
%
%
%
\begin{proof}	
	From Lemma \ref{p-CSS}, we have ${\rm PAPR}(\mathbf{\Phi}_{\mathbf{A}_{i}})\leq p$ for any $1\leq i\leq 2p$, hence, eventually we have ${\rm PAPR}(\mathbf{\Phi})\leq p$.
	
	By Theorem \ref{coherence}, the coherence of $\mathbf{\Phi}$ is depend on the minimum rank of $\mathbf{Q}_{i,j} = (\mathbf{A}_{i}-\mathbf{A}_{j})+(\mathbf{A}_{i}-\mathbf{A}_{j})^{T}$, where $1\leq i<j\leq 2p$. Recall from the proof of Theorem \ref{th_coherence_Lp}, when $i,j\in\{1,\cdots,p\}$ or $i,j\in\{p+1,\cdots,2p\}$, $\mathbf{Q}_{i,j}$ is a diagonal matrix with full rank, i.e. $rank_{p}(\mathbf{Q}_{i,j})=m$. Next, consider the case, when $i\in\{1,\cdots,p\},j\in\{p+1,\cdots,2p\}$,  we have $\mathbf{A}_{i}=\psi(\pi,\mathbf{a},\mathbf{d}_{i}),\mathbf{A}_{j}=\psi(\pi^{\tau},\mathbf{b},\mathbf{d}_{j})$, then the quadratic form corresponding to $\mathbf{Q}_{i,j}$ can be expressed as follows:
	\begin{equation}\label{psi_tat_plus_t}
		\begin{split}
			\mathbf{x}\mathbf{Q}_{i,j}\mathbf{x}^{T} &= \mathbf{x}\mathbf{A}_{i}\mathbf{x}^{T}-\mathbf{x}\mathbf{A}_{j}\mathbf{x}^{T}+\mathbf{x}\mathbf{A}_{i}^{T}\mathbf{x}^{T}-\mathbf{x}\mathbf{A}_{j}^{T}\mathbf{x}^{T}\\
			&=2 \big(\sum_{k=1}^{m-1} a_{k} x_{\pi(k)}x_{\pi(k+1)}+ \sum_{k=1}^{m} d_{i,k} x_{k}^{2}\\ & -\sum_{k=1}^{m-1} b_{k} x_{\pi^{\tau}(k)}x_{\pi^{\tau}(k+1)}-\sum_{k=1}^{m} d_{j,k} x_{k}^{2}\big)\\
			&=2\big(\sum_{k=1}^{m-1} c_{k}x_{\pi([\tau+t+k]_m)}x_{\pi([\tau+t+k+1]_m)} + \sum_{k=1}^{m} \hat{d}_{k} x_{k}^{2}\big)\\
			&=\mathbf{x}(\psi({\red{\pi^{\prime}}},\mathbf{c},\hat{\mathbf{d}})+\psi({\red{\pi^{\prime}}},\mathbf{c},\hat{\mathbf{d}})^{T})\mathbf{x}^{T},
		\end{split}	
	\end{equation}
	where {\red{$\pi^{\prime}(i)=i$ for $1\leq i\leq m$}}. The second to last identity followed by $\hat{d}_{k}=d_{i,\pi^{t+\tau}(k)}-d_{j,\pi^{t+\tau}(k)}$ and
	\begin{equation}
		c_{k}= \begin{cases}
			a_{\tau}, & \text{if~} [k+t+\tau]_{m}=\tau;\\
			-b_{m-\tau}, & \text{if~} [k+t+\tau]_{m}=m;\\
			a_{[k+t+\tau]_{m}}-b_{[k+t]_{m}}, & \text{other cases.} 
		\end{cases}
	\end{equation}
	Since $\mathbf{Q}_{i,j}=\psi({\red{\pi^{\prime}}},\mathbf{c},\hat{\mathbf{d}})+\psi({\red{\pi^{\prime}}},\mathbf{c},\hat{\mathbf{d}})^{T}$ with $\psi({\red{\pi^{\prime}}},\mathbf{c},\hat{\mathbf{d}})\in\mathcal{A}_{p}$, by Lemma \ref{lemma_rank_calA}, $rank_{p}(\mathbf{Q}_{i,j})\geq m-1$.

The proof is completed using Theorem \ref{coherence}, since $r_{min}=\min\limits_{1\leq i\not=j \leq L} rank_{p}(\mathbf{Q}_{i,j})\geq m-1$.
\end{proof}

{\red{\begin{remark}\label{number_of_phi_2}
			Using the same notation as in Remark \ref{number_of_phi}, for a given permutation $\pi$ and $\mathbf{a},\mathbf{b}\in \mathbf{F}_{p}^{m-1}$, we have
			\begin{equation}
				\mathcal{N}_{\mathbf{\Phi}}=(p!)^{2m-2}
			\end{equation}
			where $\mathbf{\Phi}$ is defined in Theorem \ref{construction2} and Theorem \ref{construction3}.
\end{remark}}}

\begin{example}\label{example1}
	Let $p=3,m=4$, $\pi=[1,2,3,4]$, $\mathbf{a}=(1,2,2),\mathbf{b} = (2,1,1)$, and $\mathbf{d}_{1}=\mathbf{d}_{4}=(0,0,0,0),\mathbf{d}_{2}=\mathbf{d}_{5}=(1,1,1,1),\mathbf{d}_{3}=\mathbf{d}_{6}=(2,2,2,2)$. According to Theorem \ref{construction2} and (\ref{eq_matrix}),  we get the following six matrices:
	
	{\small \begin{equation*}
			\begin{split}
				&	\mathbf{A}_1=\left(\begin{array}{cccc}
					0 & 1 & 0 & 0\\
					0 & 0 & 2 & 0\\
					0 & 0 & 0 & 2\\
					0 & 0 & 0 & 0
				\end{array}\right),
				\mathbf{A}_2=\left(\begin{array}{cccc}
					1 & 1 & 0 & 0\\
					0 & 1 & 2 & 0\\
					0 & 0 & 1 & 2\\
					0 & 0 & 0 & 1
				\end{array}\right),
				\mathbf{A}_3=\left(\begin{array}{cccc}
					2 & 1 & 0 & 0\\
					0 & 2 & 2 & 0\\
					0 & 0 & 2 & 2\\
					0 & 0 & 0 & 2
				\end{array}\right),\\ &\mathbf{A}_4=\left(\begin{array}{cccc}
					0 & 2 & 0 & 0\\
					0 & 0 & 1 & 0\\
					0 & 0 & 0 & 1\\
					0 & 0 & 0 & 0
				\end{array}\right),
				\mathbf{A}_5=\left(\begin{array}{cccc}
					1 & 2 & 0 & 0\\
					0 & 1 & 1 & 0\\
					0 & 0 & 1 & 1\\
					0 & 0 & 0 & 1
				\end{array}\right),
				\mathbf{A}_6=\left(\begin{array}{cccc}
					2 & 2 & 0 & 0\\
					0 & 2 & 1 & 0\\
					0 & 0 & 2 & 1\\
					0 & 0 & 0 & 2
				\end{array}\right).
			\end{split}
	\end{equation*}}
	Here, $rank_{3}(\mathbf{Q}_{i,j})=4$ for any $i\not=j$. 	
	Let $\mathbf{\Phi}=[\mathbf{\Phi}_{\mathbf{A}_1},\mathbf{\Phi}_{\mathbf{A}_2},\mathbf{\Phi}_{\mathbf{A}_3},\mathbf{\Phi}_{\mathbf{A}_4},\mathbf{\Phi}_{\mathbf{A}_5},\mathbf{\Phi}_{\mathbf{A}_6}]$. Then, ${\rm PAPR}({\mathbf{\Phi}})=\textcolor{black}{2.8738}< 3$ and $\mu(\mathbf{\Phi})=\frac{1}{3^2}=\frac{1}{9}$, which is optimum. Note that each of $\mathbf{\Phi}_{\mathbf{A}_i}$ is of size $3^4\times 3^4$. The first, second and last columns of each of the $\mathbf{\Phi}_{\mathbf{A}_i}$ for $1\leq i \leq 6$ are given in Table \ref{tabseq}, the elements are power of $\omega_3$.


\begin{table}[]
	\centering
	\caption{The spreading sequences corresponding to the matrices $\mathbf{A}_i$ for $1\leq i \leq 6$, in Example \ref{example1}\label{tabseq}.}
	\scalebox{0.7}{
		\begin{tabular}{|cccc|cccc|cccc|cccc|cccc|cccc|}
			\hline
			\multicolumn{4}{|c|}{$\mathbf{\Phi}_{\mathbf{A}_{1}}$}                                                                                                                                                                                                                                                                                                                                                                                                                                                                                                                                                                                                                                                                                                                                                                                                                                                                                                                                                                                                                                                                                       & \multicolumn{4}{c|}{$\mathbf{\Phi}_{\mathbf{A}_2}$}                                                                                                                                                                                                                                                                                                                                                                                                                                                                                                                                                                                                                                                                                                                                                                                                                                                                                                                                                                                                                                                                                          & \multicolumn{4}{c|}{$\mathbf{\Phi}_{\mathbf{A}_3}$}                                                                                                                                                                                                                                                                                                                                                                                                                                                                                                                                                                                                                                                                                                                                                                                                                                                                                                                                                                                                                                                                                          & \multicolumn{4}{c|}{$\mathbf{\Phi}_{\mathbf{A}_4}$}                                                                                                                                                                                                                                                                                                                                                                                                                                                                                                                                                                                                                                                                                                                                                                                                                                                                                                                                                                                                                                                                                          & \multicolumn{4}{c|}{$\mathbf{\Phi}_{\mathbf{A}_5}$}                                                                                                                                                                                                                                                                                                                                                                                                                                                                                                                                                                                                                                                                                                                                                                                                                                                                                                                                                                                                                                                                                          & \multicolumn{4}{c|}{$\mathbf{\Phi}_{\mathbf{A}_6}$}                                                                                                                                                                                                                                                                                                                                                                                                                                                                                                                                                                                                                                                                                                                                                                                                                                                                                                                                                                                                                                                                                          \\ \hline
			\begin{tabular}[c]{@{}c@{}}0\\0\\0\\0\\1\\2\\0\\2\\1\\0\\0\\0\\2\\0\\1\\1\\0\\2\\0\\0\\0\\1\\2\\0\\2\\1\\0\\0\\0\\0\\0\\1\\2\\0\\2\\1\\2\\2\\2\\1\\2\\0\\0\\2\\1\\1\\1\\1\\2\\0\\1\\0\\2\\1\\0\\0\\0\\0\\1\\2\\0\\2\\1\\1\\1\\1\\0\\1\\2\\2\\1\\0\\2\\2\\2\\0\\1\\2\\1\\0\\2
			\end{tabular} & \begin{tabular}[c]{@{}c@{}}0\\1\\2\\0\\2\\1\\0\\0\\0\\0\\1\\2\\2\\1\\0\\1\\1\\1\\0\\1\\2\\1\\0\\2\\2\\2\\2\\0\\1\\2\\0\\2\\1\\0\\0\\0\\2\\0\\1\\1\\0\\2\\0\\0\\0\\1\\2\\0\\2\\1\\0\\0\\0\\0\\0\\1\\2\\0\\2\\1\\0\\0\\0\\1\\2\\0\\0\\2\\1\\2\\2\\2\\2\\0\\1\\0\\2\\1\\1\\1\\1
			\end{tabular} & $\hdots$ & \begin{tabular}[c]{@{}c@{}}0\\2\\1\\2\\2\\2\\1\\2\\0\\2\\1\\0\\0\\0\\0\\1\\2\\0\\1\\0\\2\\1\\1\\1\\1\\2\\0\\2\\1\\0\\1\\1\\1\\0\\1\\2\\0\\2\\1\\1\\1\\1\\2\\0\\1\\1\\0\\2\\1\\1\\1\\1\\2\\0\\1\\0\\2\\0\\0\\0\\2\\0\\1\\1\\0\\2\\2\\2\\2\\0\\1\\2\\1\\0\\2\\1\\1\\1\\1\\2\\0
			\end{tabular} & \begin{tabular}[c]{@{}c@{}}0\\1\\1\\1\\0\\1\\1\\1\\0\\1\\2\\2\\1\\0\\1\\0\\0\\2\\1\\2\\2\\0\\2\\0\\1\\1\\0\\1\\2\\2\\2\\1\\2\\2\\2\\1\\1\\2\\2\\1\\0\\1\\0\\0\\2\\0\\1\\1\\2\\1\\2\\0\\0\\2\\1\\2\\2\\2\\1\\2\\2\\2\\1\\0\\1\\1\\0\\2\\0\\2\\2\\1\\1\\2\\2\\0\\2\\0\\1\\1\\0
				
			\end{tabular} & \begin{tabular}[c]{@{}c@{}}0\\2\\0\\1\\1\\0\\1\\2\\2\\1\\0\\1\\1\\1\\0\\0\\1\\1\\1\\0\\1\\0\\0\\2\\1\\2\\2\\1\\0\\1\\2\\2\\1\\2\\0\\0\\1\\0\\1\\1\\1\\0\\0\\1\\1\\0\\2\\0\\2\\2\\1\\0\\1\\1\\1\\0\\1\\2\\2\\1\\2\\0\\0\\0\\2\\0\\0\\0\\2\\2\\0\\0\\1\\0\\1\\0\\0\\2\\1\\2\\2
				
			\end{tabular} & $\hdots$ & \begin{tabular}[c]{@{}c@{}}0\\0\\2\\0\\1\\1\\2\\1\\2\\0\\0\\2\\2\\0\\0\\0\\2\\0\\2\\2\\1\\0\\1\\1\\0\\2\\0\\0\\0\\2\\0\\1\\1\\2\\1\\2\\2\\2\\1\\1\\2\\2\\2\\1\\2\\0\\0\\2\\1\\2\\2\\1\\0\\1\\2\\2\\1\\2\\0\\0\\1\\0\\1\\0\\0\\2\\2\\0\\0\\0\\2\\0\\0\\0\\2\\1\\2\\2\\1\\0\\1
				
			\end{tabular} & \begin{tabular}[c]{@{}c@{}}0\\2\\2\\2\\2\\0\\2\\0\\2\\2\\1\\1\\0\\0\\1\\2\\0\\2\\2\\1\\1\\2\\2\\0\\0\\1\\0\\2\\1\\1\\1\\1\\2\\1\\2\\1\\0\\2\\2\\1\\1\\2\\0\\1\\0\\2\\1\\1\\2\\2\\0\\0\\1\\0\\2\\1\\1\\1\\1\\2\\1\\2\\1\\2\\1\\1\\0\\0\\1\\2\\0\\2\\0\\2\\2\\0\\0\\1\\1\\2\\1
				
			\end{tabular} & \begin{tabular}[c]{@{}c@{}}0\\0\\1\\2\\0\\2\\2\\1\\1\\2\\2\\0\\0\\1\\0\\2\\1\\1\\2\\2\\0\\2\\0\\2\\0\\2\\2\\2\\2\\0\\1\\2\\1\\1\\0\\0\\0\\0\\1\\1\\2\\1\\0\\2\\2\\2\\2\\0\\2\\0\\2\\0\\2\\2\\2\\2\\0\\1\\2\\1\\1\\0\\0\\2\\2\\0\\0\\1\\0\\2\\1\\1\\0\\0\\1\\0\\1\\0\\1\\0\\0
				
			\end{tabular} & $\hdots$ & \begin{tabular}[c]{@{}c@{}}0\\1\\0\\1\\0\\0\\0\\0\\1\\1\\2\\1\\1\\0\\0\\2\\2\\0\\0\\1\\0\\2\\1\\1\\2\\2\\0\\1\\2\\1\\2\\1\\1\\1\\1\\2\\1\\2\\1\\1\\0\\0\\2\\2\\0\\2\\0\\2\\1\\0\\0\\1\\1\\2\\0\\1\\0\\1\\0\\0\\0\\0\\1\\2\\0\\2\\2\\1\\1\\0\\0\\1\\2\\0\\2\\1\\0\\0\\1\\1\\2
				
			\end{tabular} & \begin{tabular}[c]{@{}c@{}}0\\0\\0\\0\\2\\1\\0\\1\\2\\0\\0\\0\\1\\0\\2\\2\\0\\1\\0\\0\\0\\2\\1\\0\\1\\2\\0\\0\\0\\0\\0\\2\\1\\0\\1\\2\\1\\1\\1\\2\\1\\0\\0\\1\\2\\2\\2\\2\\1\\0\\2\\0\\1\\2\\0\\0\\0\\0\\2\\1\\0\\1\\2\\2\\2\\2\\0\\2\\1\\1\\2\\0\\1\\1\\1\\0\\2\\1\\2\\0\\1
				
			\end{tabular} & \begin{tabular}[c]{@{}c@{}}0\\1\\2\\0\\0\\0\\0\\2\\1\\0\\1\\2\\1\\1\\1\\2\\1\\0\\0\\1\\2\\2\\2\\2\\1\\0\\2\\0\\1\\2\\0\\0\\0\\0\\2\\1\\1\\2\\0\\2\\2\\2\\0\\2\\1\\2\\0\\1\\1\\1\\1\\0\\2\\1\\0\\1\\2\\0\\0\\0\\0\\2\\1\\2\\0\\1\\0\\0\\0\\1\\0\\2\\1\\2\\0\\0\\0\\0\\2\\1\\0
				
			\end{tabular} & $\hdots$ & \begin{tabular}[c]{@{}c@{}}0\\2\\1\\2\\0\\1\\1\\1\\1\\2\\1\\0\\2\\0\\1\\2\\2\\2\\1\\0\\2\\2\\0\\1\\0\\0\\0\\2\\1\\0\\1\\2\\0\\0\\0\\0\\2\\1\\0\\2\\0\\1\\2\\2\\2\\2\\1\\0\\0\\1\\2\\1\\1\\1\\1\\0\\2\\0\\1\\2\\2\\2\\2\\2\\1\\0\\2\\0\\1\\2\\2\\2\\0\\2\\1\\1\\2\\0\\2\\2\\2
				
			\end{tabular} & \begin{tabular}[c]{@{}c@{}}0\\1\\1\\1\\1\\0\\1\\0\\1\\1\\2\\2\\0\\0\\2\\1\\0\\1\\1\\2\\2\\1\\1\\0\\0\\2\\0\\1\\2\\2\\2\\2\\1\\2\\1\\2\\0\\1\\1\\2\\2\\1\\0\\2\\0\\1\\2\\2\\1\\1\\0\\0\\2\\0\\1\\2\\2\\2\\2\\1\\2\\1\\2\\1\\2\\2\\0\\0\\2\\1\\0\\1\\0\\1\\1\\0\\0\\2\\2\\1\\2
				
			\end{tabular} & \begin{tabular}[c]{@{}c@{}}0\\2\\0\\1\\2\\2\\1\\1\\0\\1\\0\\1\\0\\1\\1\\1\\1\\0\\1\\0\\1\\1\\2\\2\\0\\0\\2\\1\\0\\1\\2\\0\\0\\2\\2\\1\\0\\2\\0\\2\\0\\0\\0\\0\\2\\1\\0\\1\\1\\2\\2\\0\\0\\2\\1\\0\\1\\2\\0\\0\\2\\2\\1\\1\\0\\1\\0\\1\\1\\1\\1\\0\\0\\2\\0\\0\\1\\1\\2\\2\\1
				
			\end{tabular} & $\hdots$ & \begin{tabular}[c]{@{}c@{}}0\\0\\2\\0\\2\\0\\2\\0\\0\\0\\0\\2\\1\\0\\1\\1\\2\\2\\2\\2\\1\\1\\0\\1\\2\\0\\0\\0\\0\\2\\0\\2\\0\\2\\0\\0\\1\\1\\0\\2\\1\\2\\2\\0\\0\\1\\1\\0\\0\\2\\0\\1\\2\\2\\2\\2\\1\\2\\1\\2\\1\\2\\2\\1\\1\\0\\2\\1\\2\\2\\0\\0\\2\\2\\1\\1\\0\\1\\2\\0\\0
				
			\end{tabular} & \begin{tabular}[c]{@{}c@{}}0\\2\\2\\2\\0\\2\\2\\2\\0\\2\\1\\1\\2\\0\\2\\0\\0\\1\\2\\1\\1\\0\\1\\0\\2\\2\\0\\2\\1\\1\\1\\2\\1\\1\\1\\2\\2\\1\\1\\2\\0\\2\\0\\0\\1\\0\\2\\2\\1\\2\\1\\0\\0\\1\\2\\1\\1\\1\\2\\1\\1\\1\\2\\0\\2\\2\\0\\1\\0\\1\\1\\2\\2\\1\\1\\0\\1\\0\\2\\2\\0
				
			\end{tabular} & \begin{tabular}[c]{@{}c@{}}0\\0\\1\\2\\1\\1\\2\\0\\2\\2\\2\\0\\2\\1\\1\\0\\1\\0\\2\\2\\0\\0\\2\\2\\2\\0\\2\\2\\2\\0\\1\\0\\0\\1\\2\\1\\2\\2\\0\\2\\1\\1\\0\\1\\0\\0\\0\\1\\1\\0\\0\\0\\1\\0\\2\\2\\0\\1\\0\\0\\1\\2\\1\\0\\0\\1\\0\\2\\2\\1\\2\\1\\2\\2\\0\\0\\2\\2\\2\\0\\2
				
			\end{tabular} & $\hdots$ & \begin{tabular}[c]{@{}c@{}}0\\1\\0\\1\\1\\2\\0\\2\\2\\1\\2\\1\\0\\0\\1\\0\\2\\2\\0\\1\\0\\0\\0\\1\\1\\0\\0\\1\\2\\1\\2\\2\\0\\1\\0\\0\\0\\1\\0\\2\\2\\0\\2\\1\\1\\0\\1\\0\\0\\0\\1\\1\\0\\0\\0\\1\\0\\1\\1\\2\\0\\2\\2\\0\\1\\0\\2\\2\\0\\2\\1\\1\\1\\2\\1\\1\\1\\2\\2\\1\\1\end{tabular} \\ \hline
	\end{tabular}}
\end{table}
\end{example}

{\red{\begin{example}\label{example_4_3}
			Let $p=3,m=4$, $\pi=[2,3,1,4]$, $\mathbf{a}=(1,1,2),\mathbf{b}=(1,1,1)$, $\tau=1$ and $\mathbf{d}_{1}=(0,2,1,2),\mathbf{d}_{2}=(1,1,0,0),\mathbf{d}_{3}=(2,0,2,1),\mathbf{d}_{4}=(0,0,0,0),\mathbf{d}_{5}=(1,1,1,1),\mathbf{d}_{6}=(2,2,2,2)$. According to Theorem \ref{construction3} and (\ref{eq_matrix}), we get the following six matrices:
			\begin{equation*}
				\begin{split}
					&	\mathbf{A}_1=\left(\begin{array}{cccc}
						0 & 0 & 0 & 2\\
						0 & 2 & 1 & 0\\
						1 & 0 & 1 & 0\\
						0 & 0 & 0 & 2
					\end{array}\right),
					\mathbf{A}_2=\left(\begin{array}{cccc}
						1 & 0 & 0 & 2\\
						0 & 1 & 1 & 0\\
						1 & 0 & 0 & 0\\
						0 & 0 & 0 & 0
					\end{array}\right),
					\mathbf{A}_3=\left(\begin{array}{cccc}
						2 & 0 & 0 & 2\\
						0 & 0 & 1 & 0\\
						1 & 0 & 2 & 0\\
						0 & 0 & 0 & 1
					\end{array}\right),\\ &\mathbf{A}_4=\left(\begin{array}{cccc}
						0 & 0 & 0 & 1\\
						0 & 0 & 0 & 0\\
						1 & 0 & 0 & 0\\
						0 & 1 & 0 & 0
					\end{array}\right),
					\mathbf{A}_5=\left(\begin{array}{cccc}
						1 & 0 & 0 & 1\\
						0 & 1 & 0 & 0\\
						1 & 0 & 1 & 0\\
						0 & 1 & 0 & 1
					\end{array}\right),
					\mathbf{A}_6=\left(\begin{array}{cccc}
						2 & 0 & 0 & 1\\
						0 & 2 & 0 & 0\\
						1 & 0 & 2 & 0\\
						0 & 1 & 0 & 2
					\end{array}\right).
				\end{split}
			\end{equation*}
			Here, $rank_{3}(\mathbf{Q}_{i,j})=4$ for any $i\not=j$. 	
			Let $\mathbf{\Phi}=\frac{1}{3^{4}}[\mathbf{\Phi}_{\mathbf{A}_1},\mathbf{\Phi}_{\mathbf{A}_2},\mathbf{\Phi}_{\mathbf{A}_3},\mathbf{\Phi}_{\mathbf{A}_4},\mathbf{\Phi}_{\mathbf{A}_5},\mathbf{\Phi}_{\mathbf{A}_6}]$ be the  $3$-ary spreading matrix of size $3^{4}\times (2\cdot3^{5})$. Then, ${\rm PAPR}({\mathbf{\Phi}})=2.8795<3$ and $\mu(\mathbf{\Phi})=\sqrt{\frac{1}{3^{4}}}\approx0.1111<\sqrt{\frac{3}{3^{4}}}$, which is optimum.
\end{example}}}
\subsection{Proposed infinite families of spreading sequence sets $\mathbf{\Phi}$ with optimum coherence for $p=3$ \label{sec4_2}}
In the previous sub-section we propose several infinite families of spreading sequence matrices $\mathcal{M}$ with $r_{min}\geq m-1$. In this section, we are focused on constructing classes of $\mathcal{M}$ with $r_{min}=m$, which can be used to generate $\mathbf{\Phi}$ with optimum coherence. Based on the above constructions, several classes of $\mathcal{M}$ with $r_{min}=m$ will be proposed for $p=3$. 

\begin{theorem}\label{construction_optimum_p=3_m_even_1}
	Let $p=3$, $m\geq 4$ {\red{be}} an even integer with $m\not\equiv 2 \bmod 3$ and $\pi$ {\red{be}} a permutation over $\{1,\cdots,m\}$. Let $\mathbf{d}_{1}$ be a vector defined in $\mathbb{F}_{3}^{m}$, and for $i=2,3$ define $\mathbf{d}_{i}=\mathbf{d}_{1}+(i-1)\mathbf{1}_m$. Let $\mathbf{a},\mathbf{b}\in\mathbb{F}_{3}^{m-1}$ such that $\forall k\in\{1,\cdots,m-1\},a_{k},b_{k}\not=0, a_{k}-b_{k}\not=0$. Let $\mathcal{M} = \{\mathbf{A}_{1},\mathbf{A}_{2},\mathbf{A}_{3},\mathbf{A}_{4},\mathbf{A}_{5},\mathbf{A}_{6}\}$ with
	\begin{equation}
		\mathbf{A}_{i}=\begin{cases}
			\psi(\pi,\mathbf{a},\mathbf{d}_{i}), & 1\leq i\leq 3;\\
			\psi(\pi,\mathbf{b},\mathbf{d}_{i-3}), & 4\leq i\leq 6.
		\end{cases}
	\end{equation}
	Let $\mathbf{\Phi}$ be the spreading sequence set generated by $\mathcal{M}$, as defined in (\ref{phimatrix}). Then {\red{$\mathbf{\Phi}$ is a sequence set with $2\times 3^{m+1}$ sequences, each of length $3^m$, and}} ${\rm PAPR}(\mathbf{\Phi})\leq 3$, $\mu(\mathbf{\Phi})=\sqrt{\frac{1}{M}}$.
\end{theorem}
\begin{proof}
	From Lemma \ref{p-CSS}, we have ${\rm PAPR}(\mathbf{\Phi}_{\mathbf{A}_{i}})\leq 3$ for any $1\leq i\leq 6$, hence, eventually we have ${\rm PAPR}(\mathbf{\Phi})\leq 3$.

	Next, let us calculate the coherence. Since $\mathbf{A}_{i}-\mathbf{A}_{j}$ with $i,j\in\{1,2,3\}$ or $i,j\in \{4,5,6\}$ is a diagonal matrix with full rank, we only need to consider the rank of $\mathbf{Q}_{i,j}$ for $i\in\{1,2,3\},j\in\{4,5,6\}$.
	
	By Theorem \ref{construction2},  if $i\in\{1,2,3\},j\in\{4,5,6\}$, we have $\mathbf{A}_{i}-\mathbf{A}_{j}=\psi(\pi,\mathbf{c},(i-j)\mathbf{1}_{m})$ where $\mathbf{c}=\mathbf{a}-\mathbf{b}$. By Lemma \ref{lemma_rank_calA}, the rank of $\mathbf{Q}_{i,j}$ is equal to the rank of the following tridiagonal matrix:
	
	\begin{equation}
		\mathbf{D} = \left(\begin{array}{cccccc}
			2(i-j) & c_{1} & & & & \\
			c_{1} & 2(i-j) & c_{2} & & & \\
			& c_{2} & 2(i-j) & c_{3} & & \\
			& & \ddots & \ddots & \ddots & \\
			& & & c_{m-2} & 2(i-j) & c_{m-1} \\
			& & & & c_{m-1} & 2(i-j)
		\end{array}\right),
	\end{equation}	
where $\mathbf{D}=\mathbf{B}_{(i,j)}+\mathbf{B}_{(i,j)}^T$ and $\mathbf{B}_{i,j}=\psi(\pi^\prime,\mathbf{c},(i-j)\mathbf{1}_{m}) \text{ where } \pi^\prime(i)=i$ for $1\leq i \leq m$.
	Since $m\not\equiv 2\bmod 3$ and $i-j\in{\red{\mathbb{F}_{3}^{*}}}$, {\red{we obtain}} $(i-j)^2=c_{i}^{2}=1$. By simple calculations, we get $\det(\mathbf{D})\not=0$, which implies that $rank_{p}(\mathbf{Q}_{i,j})=m$ for $i\in\{1,2,3\},j\in\{4,5,6\}$. Hence, we conclude that $r_{min} = m$. 
	
	Therefore, using Theorem \ref{coherence}, $\mu(\mathbf{\Phi})=\sqrt{\frac{1}{M}}$.
\end{proof}
%

{\red{\begin{example}\label{example_4_4}
			Let $p=3,m=4$, $\pi=[1,4,3,2]$, $\mathbf{a}=(2,2,2),\mathbf{b}=(1,1,1)$, and $\mathbf{d}_{1}=(0,0,0,0),\mathbf{d}_{2}=(1,1,1,1),\mathbf{d}_{3}=(2,2,2,2),\mathbf{d}_{4}=(0,0,1,0),\mathbf{d}_{5}=(1,1,2,1),\mathbf{d}_{6}=(2,2,0,2)$. According to Theorem \ref{construction_optimum_p=3_m_even_1} and (\ref{eq_matrix}), we get the following $6$ matrices:
			\begin{equation*}
				\begin{split}
					&	\mathbf{A}_1=\left(\begin{array}{cccc}
						0 & 0 & 0 & 2\\
						0 & 0 & 0 & 0\\
						0 & 2 & 0 & 0\\
						0 & 0 & 2 & 0
					\end{array}\right),
					\mathbf{A}_2=\left(\begin{array}{cccc}
						1 & 0 & 0 & 2\\
						0 & 1 & 0 & 0\\
						0 & 2 & 1 & 0\\
						0 & 0 & 2 & 1
					\end{array}\right),
					\mathbf{A}_3=\left(\begin{array}{cccc}
						2 & 0 & 0 & 2\\
						0 & 2 & 0 & 0\\
						0 & 2 & 2 & 0\\
						0 & 0 & 2 & 2
					\end{array}\right),\\ &\mathbf{A}_4=\left(\begin{array}{cccc}
						0 & 0 & 0 & 1\\
						0 & 0 & 0 & 0\\
						0 & 1 & 1 & 0\\
						0 & 0 & 1 & 0
					\end{array}\right),
					\mathbf{A}_5=\left(\begin{array}{cccc}
						1 & 0 & 0 & 1\\
						0 & 1 & 0 & 0\\
						0 & 1 & 2 & 0\\
						0 & 0 & 1 & 1
					\end{array}\right),
					\mathbf{A}_6=\left(\begin{array}{cccc}
						2 & 0 & 0 & 1\\
						0 & 2 & 0 & 0\\
						0 & 1 & 0 & 0\\
						0 & 0 & 1 & 2
					\end{array}\right).
				\end{split}
			\end{equation*}
			Here, $rank_{3}(\mathbf{Q}_{i,j})=4$ for any $i\not=j$. 	
			Let $\mathbf{\Phi}=\frac{1}{3^{4}}[\mathbf{\Phi}_{\mathbf{A}_1},\mathbf{\Phi}_{\mathbf{A}_2},\mathbf{\Phi}_{\mathbf{A}_3},\mathbf{\Phi}_{\mathbf{A}_4},\mathbf{\Phi}_{\mathbf{A}_5},\mathbf{\Phi}_{\mathbf{A}_6}]$ be the  $3$-ary spreading matrix of size $3^{4}\times (2\cdot3^{5})$. Then, ${\rm PAPR}({\mathbf{\Phi}})=2.8931<3$ and $\mu(\mathbf{\Phi})=\sqrt{\frac{1}{3^{4}}}\approx0.1111<\sqrt{\frac{3}{3^{4}}}$, which is optimum.
\end{example}}}

The above result presents a family of $\mathbf{\Phi}$ with optimum coherence for $p=3$, however $m$ is constrained to an even integer with $m\not\equiv 2\bmod 3$. Inspired by Theorem \ref{construction_optimum_p=3_m_even_1}, we find a class of $\mathcal{M}$ with $r_{min}=m$ for $p=3$ while $m\geq 2$ can be any integer. 
\begin{theorem}\label{construction_optimum_p=3_2}
	Let $p=3$ and $m=3l+u$ {\red{be}} a positive integer, $u\in\{0,1,2\}$. Let $\pi$ be a permutation over $\{1,\cdots,m\}$ and $\mathbf{d}_{1}\in\mathbb{F}_{3}^{m}$.  For $i=2,3$, define $\mathbf{d}_{i}=\mathbf{d}_{1}+(i-1)\mathbf{1}_m$. Suppose that $U$ is a index set defined as follows:
	\begin{equation}
		U=\begin{cases}
			U_{0}\cup U_{1}, &\text{if $u=0$}\\
			U_{0}\cup U_{2}
			, &\text{if $u=1$}\\
			U_{1}\cup U_{2}, &\text{if $u=2$}\\
		\end{cases}
	\end{equation}
	where
	\begin{equation}
		U_{i}=\begin{cases}
			\{l_1\mid 1\leq l_{1} \leq m,l_1 \equiv i \bmod 3\},& \text{$m$ is even};\\
			\{l_1\mid 1\leq l_{1} \leq m,l_1 \equiv i \bmod 3,l_1\equiv 1\bmod 2\},& \text{$m$ is odd}.		
		\end{cases}
	\end{equation}
	Let $s\in U$ with $\pi(s)=s'$, and $e\in\mathbb{F}_{3}\setminus\{0\}$, for $i\in\{4,5,6\}$ define $\mathbf{d}_{i}$ as follows:
	\begin{equation}\label{condition_2}
		d_{i,k}=\begin{cases}
			d_{i-3,k}, &k\in\{1,\cdots,m\}\setminus\{s'\};\\
			d_{i-3,k}+e,& k=s'.
		\end{cases}
	\end{equation}
	Let $\mathbf{a},\mathbf{b}$ satisfy the condition that $\forall k\in\{1,\cdots,m-1\}, a_{k},b_{k}\not=0,a_{k}\not=b_{k}$. Let $\mathcal{M}=\{\mathbf{A}_{1},\mathbf{A}_{2},\mathbf{A}_{3},\mathbf{A}_{4},\mathbf{A}_{5},\mathbf{A}_{6}\}$ with
	\begin{equation}
		\mathbf{A}_{i}=\begin{cases}
			\psi(\pi,\mathbf{a},\mathbf{d}_{i}), & 1\leq i\leq 3;\\
			\psi(\pi,\mathbf{b},\mathbf{d}_{i}), & 4\leq i\leq 6.
		\end{cases}
	\end{equation}
Let $\mathbf{\Phi}$ be the spreading sequence set generated by $\mathcal{M}$, as defined in (\ref{phimatrix}). Then {\red{$\mathbf{\Phi}$ is a sequence set with $2\times3^m$ sequences, each of length $3^m$, and}} ${\rm PAPR}(\mathbf{\Phi})\leq 3$, $\mu(\mathbf{\Phi})=\sqrt{\frac{1}{M}}$.
\end{theorem}

\begin{proof}
	From Lemma \ref{p-CSS}, we have ${\rm PAPR}(\mathbf{\Phi}_{\mathbf{A}_{i}})\leq 3$ for any $1\leq i\leq 6$, hence, eventually we have ${\rm PAPR}(\mathbf{\Phi})\leq 3$.
	
	Similar to the proof of Theorem \ref{construction_optimum_p=3_m_even_1}, we only need to consider the rank of $\mathbf{Q}_{i,j}$ for $i\in\{1,2,3\},j\in\{4,5,6\}$.
	
	For $j=i+3$, we have $\mathbf{Q}_{i,j}=\psi(\pi,\mathbf{c},\mathbf{d}_{i}-\mathbf{d}_{j})+\psi(\pi,\mathbf{c},\mathbf{d}_{i}-\mathbf{d}_{j})^{T}$ where $\mathbf{c}=\mathbf{a}-\mathbf{b}$. By Lemma \ref{lemma_rank_calA}, the rank of $rank_p(\mathbf{Q}_{i,j})=rank_p(\mathbf{D}_{i,j})$, where $\mathbf{D}_{i,j}=\psi(\pi^\prime,\mathbf{c},\mathbf{d}^{'})+\psi(\pi^\prime,\mathbf{c},  \mathbf{d}^{'})^{T}$,  $\pi^\prime(i)=i$ for $1\leq i \leq m$ and the $k$-th element of $\mathbf{d}^{'}$ is given by $d^{'}_{k}=d_{i,\pi(k)}-d_{j,\pi(k)}$. As per the construction, the matrix $\mathbf{D}_{i,j}$ is as follows:
	\begin{equation}
		\mathbf{D}_{i,j}=\left(\begin{array}{cccccc}
			2d_{1}^{'} & c_{1} & & & & \\
			c_{1} & 2d_{2}^{'} & c_{2} & & & \\
			& c_{2} & 2d_{3}^{'} & c_{3} & & \\
			& & \ddots & \ddots & \ddots & \\
			& & & c_{m-2} & 2d_{m-1}^{'} & c_{m-1} \\
			& & & & c_{m-1} & 2d_{m}^{'}
		\end{array}\right).
	\end{equation}
From (\ref{condition_2}), we have $d_{k}^{'}=d_{i,\pi(k)}-d_{j-3,\pi(k)}=0$ for $\pi(k)\not=s$ and $d_{s}^{'}=-e$, where $s\in U$, then	 
\begin{equation}
	\mathbf{D}_{i,j}=\left(\begin{array}{ccccccc}
		0 & c_{1} & & & & & \\
		c_{1} & \ddots & \ddots & & & &\\
		& \ddots & 0 & c_{s-1} & & &\\
		& & c_{s-1} & -2e & c_{s} & &\\
		& & & c_{s} & 0 & \ddots &\\
		& & & & \ddots & \ddots & c_{m-1} \\
		& & & &  & c_{m-1} &  0
	\end{array}\right).
\end{equation}
Calculating the determinant of $\mathbf{D}_{i,j}$, we have:
\begin{equation}
	\det(\mathbf{D}_{i,j})=\begin{cases}
		\prod\limits_{k=1}^{\frac{m}{2}} (-c_{2k-1}^{2}), & m\equiv 0\bmod 2;\\
		(-2e)\times \prod\limits_{k_{1}=1}^{\frac{s-1}{2}} (-c_{2k_{1}-1}^{2}) \times \prod\limits_{k_{2}=\frac{s+1}{2}}^{\frac{m-1}{2}}  (-c_{2k_{2}}^{2}), & m\equiv 1\bmod 2.
	\end{cases}
\end{equation}
Since $\forall k\in\{1,\cdots,m\},c_{k}=a_{k}-b_{k}\not=0$ and $e\not=0$, we have $\det(\mathbf{D}_{i,j})\not=0$, and hence, $rank_{p}(\mathbf{Q}_{i,j})=rank_{p}(\mathbf{D}_{i,j})=m$.

Similarly, for $i+3\not=j$, suppose that $j^{'}+3=j$, we obtain $\mathbf{Q}_{i,j} = \psi(\pi,\mathbf{c},\mathbf{d}_{i}-\mathbf{d}_{j^{'}+3})+\psi(\pi,\mathbf{c},\mathbf{d}_{i}-\mathbf{d}_{j^{'}+3})^{T} $, where $\mathbf{c}=\mathbf{a}-\mathbf{b}$. Suppose that $\mathbf{D}_{i,j} = \psi(\pi^\prime,\mathbf{c},\mathbf{d}^{'})+\psi(\pi^\prime,\mathbf{c},  \mathbf{d}^{'})^{T}$, where $\pi^\prime(i)=i$ for $1\leq i \leq m$ and $d^{'}_{k}=d_{i,\pi(k)}-d_{j,\pi(k)}$. By Lemma \ref{lemma_rank_calA}, we have $rank_{p}(\mathbf{Q}_{i,j})=rank_{p}(\mathbf{D}_{i,j})$. In this case, $\mathbf{D}_{i,j}$ is as follows:

\begin{equation}
	\mathbf{D}_{i,j}=\left(\begin{array}{cccccc}
		2d_{1}^{'} & c_{1} & & & & \\
		c_{1} & 2d_{2}^{'} & c_{2} & & & \\
		& c_{2} & 2d_{3}^{'} & c_{3} & & \\
		& & \ddots & \ddots & \ddots & \\
		& & & c_{m-2} & 2d_{m-1}^{'} & c_{m-1} \\
		& & & & c_{m-1} & 2d_{m}^{'}
	\end{array}\right).
\end{equation}
From (\ref{condition_2}), we have $d_{k}^{'}=d_{i,\pi(k)}-d_{j,\pi(k)}=d_{i,\pi(k)}-d_{j^{'}+3,\pi(k)}=i-j'$ for $\pi(k)\not=s$ and $d_{s}^{'}=i-j'-e$, where $s\in U$, then	 
\begin{equation}
	\mathbf{D}_{i,j}=\left(\begin{array}{ccccccc}
		2(i-j^{'}) & c_{1} & & & & & \\
		c_{1} & \ddots & \ddots & & & &\\
		& \ddots & 2(i-j^{'}) & c_{s-1} & & &\\
		& & c_{s-1} & 2(i-j^{'})-2e & c_{s} & &\\
		& & & c_{s} & 2(i-j^{'}) & \ddots &\\
		& & & & \ddots & \ddots & c_{m-1} \\
		& & & &  & c_{m-1} &  2(i-j^{'})
	\end{array}\right).
\end{equation}
Therefore,
\begin{equation}\label{eq48}
	\begin{split}
		\det(\mathbf{D}_{i,j}) &= 2d_{1}^{'}\det(\mathbf{D}_{i,j}^{(m-1\times m-1)}) + (-c_{1}^{2})\det(\mathbf{D}_{i,j}^{(m-2\times m-2)})\\
		&= ((2d_{1}{'})(2d_{2}^{'})-c_{1}^{2})\det(\mathbf{D}_{i,j}^{(m-2\times m-2)}) + (2d^{'})(-c_{2}^{2})\det(\mathbf{D}_{i,j}^{(m-3\times m-3)})\\
		& = (2d_{1}^{'})(-c_{2}^{2})\det(\mathbf{D}_{i,j}^{(m-3\times m-3)})
	\end{split}
\end{equation}
where $\mathbf{D}_{i,j}^{(m-1)\times (m-1)}$ denotes the bottom right sub-matrix of order $(m-1)$ of $\mathbf{D}_{i,j}$. Since $((2d_{1}{'})(2d_{2}^{'})-c_{1}^{2}) = (2(i-j^{'}))^{2}-c_{1}^{2} = 0$ with $i-j^{'},c_{1} \in \mathbb{F}_{3}\setminus\{0\}$, (\ref{eq48}) holds. Suppose that $s=3t+v$ with $0\leq v\leq 2$, by applying the above method, we get:

\begin{equation}
	\begin{split}
		&\det(\mathbf{D}_{i,j})=\\
		&\begin{cases}
			(2(i-j^{'}))^{l}(-c_{s+v-1}^{2})\prod\limits_{k_{1}=0}^{t+v-2} (-c_{3k_{1}+2}^{2}) \times \prod\limits_{k_{2}=t+v}^{l-1} (-c_{3k_{2}+1}^{2}), \\ \hspace{9.5cm} m\equiv 0\bmod 3;\\
			(2(i-j^{'}))^{l+1}(-c_{s-1+\lceil\frac{v}{3}\rceil}^{2})\prod\limits_{k_{1}=0}^{t+\lceil\frac{v}{3}\rceil-2} (-c_{3k_{1}+2}^{2}) \times \prod\limits_{k_{2}=t+\lceil\frac{v}{3}\rceil}^{l-1} (-c_{3k_{2}+2}^{2}), \\ \hspace{9.5cm} m\equiv 1\bmod 3;\\
			((i-j^{'}-e)(i-j^{'})-c_{s-v+1}^{2})(2(i-j^{'}))^{l}\prod\limits_{k_{1}=0}^{t-1} (-c_{3k_{1}+1}^{2}) \times \prod\limits_{k_{2}=t+1}^{l-1} (-c_{3k_{2}+3}^{2}), \\ \hspace{9.5cm} m\equiv 2\bmod 3;
		\end{cases}
	\end{split}
\end{equation}
Since $\forall k\in\{1,\cdots,m\},c_{k}=a_{k}-b_{k}\not=0$ and $i-j^{'}\not=0$, we have $\det(\mathbf{D}_{i,j})\not=0$, and hence, $rank_{p}(\mathbf{Q}_{i,j})=rank_{p}(\mathbf{D}_{i,j})=m$.
The rest of the proof is completed using Theorem \ref{coherence}, since $r_{min}=m$ and $M=p^m$. 
\end{proof}

{\red{\begin{example}\label{example_4_5}
			Let $p=3,m=5$, $\pi=[1,2,3,4,5]$, $\mathbf{a}=(2,1,2,2),\mathbf{b}=(1,2,1,1)$, and $\mathbf{d}_{1}=(0,0,0,0,0),\mathbf{d}_{2}=(1,1,1,1,1),\mathbf{d}_{3}=(2,2,2,2,2),\mathbf{d}_{4}=(0,0,0,0,1),\mathbf{d}_{5}=(1,1,1,1,2),\mathbf{d}_{6}=(2,2,2,2,0)$. According to Theorem \ref{construction_optimum_p=3_2} and (\ref{eq_matrix}), we get the following $6$ matrices:
			\begin{equation*}
				\begin{split}
					&	\mathbf{A}_1=\left(\begin{array}{ccccc}
						0 & 2 & 0 & 0 & 0 \\
						0 & 0 & 1 & 0 & 0 \\
						0 & 0 & 0 & 2 & 0 \\
						0 & 0 & 0 & 0 & 2 \\
						0 & 0 & 0 & 0 & 0 
					\end{array}\right),
					\mathbf{A}_2=\left(\begin{array}{ccccc}
						1 & 2 & 0 & 0 & 0 \\
						0 & 1 & 1 & 0 & 0 \\
						0 & 0 & 1 & 2 & 0 \\
						0 & 0 & 0 & 1 & 2 \\
						0 & 0 & 0 & 0 & 1 
					\end{array}\right),
					\mathbf{A}_3=\left(\begin{array}{ccccc}
						2 & 2 & 0 & 0 & 0 \\
						0 & 2 & 1 & 0 & 0 \\
						0 & 0 & 2 & 2 & 0 \\
						0 & 0 & 0 & 2 & 2 \\
						0 & 0 & 0 & 0 & 2 
					\end{array}\right),\\ &\mathbf{A}_4=\left(\begin{array}{ccccc}
						0 & 1 & 0 & 0 & 0 \\
						0 & 0 & 2 & 0 & 0 \\
						0 & 0 & 0 & 1 & 0 \\
						0 & 0 & 0 & 0 & 1 \\
						0 & 0 & 0 & 0 & 1 
					\end{array}\right),
					\mathbf{A}_5=\left(\begin{array}{ccccc}
						1 & 1 & 0 & 0 & 0 \\
						0 & 1 & 2 & 0 & 0 \\
						0 & 0 & 1 & 1 & 0 \\
						0 & 0 & 0 & 1 & 1 \\
						0 & 0 & 0 & 0 & 2 
					\end{array}\right),
					\mathbf{A}_6=\left(\begin{array}{ccccc}
						2 & 1 & 0 & 0 & 0 \\
						0 & 2 & 2 & 0 & 0 \\
						0 & 0 & 2 & 1 & 0 \\
						0 & 0 & 0 & 2 & 1 \\
						0 & 0 & 0 & 0 & 0 
					\end{array}\right).
				\end{split}
			\end{equation*}
			Here, $rank_{3}(\mathbf{Q}_{i,j})=5$ for any $i\not=j$. 	
			Let $\mathbf{\Phi}=\frac{1}{3^{5}}[\mathbf{\Phi}_{\mathbf{A}_1},\mathbf{\Phi}_{\mathbf{A}_2},\mathbf{\Phi}_{\mathbf{A}_3},\mathbf{\Phi}_{\mathbf{A}_4},\mathbf{\Phi}_{\mathbf{A}_5},\mathbf{\Phi}_{\mathbf{A}_6}]$ be the $3$-ary spreading matrix of size $3^{5}\times (2\cdot3^{6})$. Then, ${\rm PAPR}({\mathbf{\Phi}})=3$ and $\mu(\mathbf{\Phi})=\sqrt{\frac{1}{3^{5}}}\approx0.0642$, which is optimum.
\end{example}}}

\subsection{\red{Spreading sequence sets with low coherence and low PAPR over alphabet size $p^h$}}
In \cite{Nam2020}, two classes of $2^h$-ary spreading sequences were proposed by using the quadratic forms defined in $\mathbb{F}_{2}$. Similarly, we can extend {\red{our}} construction to $p^h$-ary case by applying the quadratic matrices defined in $\mathbb{F}_{p}$ and make the parameters more flexible.

	For any $m\geq 2$, let $q=p^h$ for $1\leq h \leq m$. Define 
	\begin{equation}
		\mathcal{L}_{c}(\mathbf{x}) = \frac{q}{p} \sum_{i=h+1}^{m} v_{i}x_{i} + d\sum_{i=1}^{h} x_{i}p^{i-1},
	\end{equation} 
	where $c=d+\sum_{i=h+1}^{m}v_{i}p^{i-1}$ for $v_{i}\in \mathbb{Z}_{p}$ and $d\in \mathbb{Z}_{q}$. {\red{Using $\mathcal{L}_{c}(\mathbf{x})$ as a linear form, similar to (\ref{eq_quadratic_functions}), we define an EBF $f$ from $\mathbb{F}_{p}^{m}$ to $\mathbb{Z}_{q}$ as follows:
	\begin{equation}\label{eq_q_ary_function}
		f(\mathbf{x})=\frac{q}{p}\left(\sum_{i=1}^{m-1} a_{i}x_{\pi(i)}x_{\pi(i+1)}+\sum_{k=1}^{m} d_{k} x_{k}^{2}\right) + \mathcal{L}_{c}(x) = \frac{q}{p} \mathbf{x}^{T}\mathbf{A}\mathbf{x} + \mathcal{L}_{c}(\mathbf{x})
	\end{equation} 
	where $\mathbf{A}\in \mathcal{A}_{p}$ is a square matrix.}}
	{\red{\begin{lemma}\label{lemma_q_ary_CS}
		Let $p$ be a prime, $m\geq 2$ be a positive integer, and $q=p^h$ with $1\leq h\leq m$. For a given $\mathbf{A}\in \mathcal{A}_{p}\subset M_{m}(\mathbb{F}_{p})$ and $0\leq c\leq p^m-1$, define the EBFs $f_{0},f_{1},\cdots,f_{p-1}$ from $\mathbb{F}_{p}^{m}$ to $\mathbb{Z}_{q}$ as follows:
		\begin{equation}
			\begin{split}
			f_{0}(\mathbf{x})&=\frac{q}{p}\mathbf{x}^{T}\mathbf{A}\mathbf{x}+\mathcal{L}_{c}(\mathbf{x}),\\
			f_{1}(\mathbf{x})&=f_{0}(\mathbf{x})+\frac{q}{p} x_{\pi(1)},\\
			&\vdots\\
			f_{p-1}(\mathbf{x})&=f_{0}(\mathbf{x})+\frac{q}{p} (p-1) x_{\pi(1)}.
		\end{split}
		\end{equation} 
		then $\{f_{0},\cdots,f_{p-1}\}$ forms a $q$-ary $(p,p^m)$-CS.
	\end{lemma}
	\begin{proof}
		Suppose that $\mathbf{s}_{n}$ denotes the corresponding $p^h$-ary complex valued sequence of $f_{n}$($0\leq n\leq p-1$). For a given $\tau\neq 0$, by (\ref{eq_sum_of_autocorrelation_of_CS}), we have:
		\begin{equation}
			\begin{split}
				\sum_{n=0}^{p-1} R_{\mathbf{s}_{n}}(\tau)& = \sum_{n=0}^{p-1}\sum_{i=0}^{p^m-1-\tau} \omega_{q}^{f_{n,i}-f_{n,i+\tau}}\\
				&=\sum_{i=0}^{p^m-1-\tau}\sum_{n=0}^{p-1} \omega_{q}^{f_{n,i}-f_{n,j}}(\text{let $j=i+\tau$}),
			\end{split}
		\end{equation}
		where $f_{n,i}=f_{n}(i_{1},\cdots,i_{m})$ with $\sum_{k=1}^{m} i_{k}p^{k-1}=i$. 
		
		Since 
		\begin{equation}
			f_{n,i}-f_{n,j} = f_{0,i}-f_{0,j} + \frac{q}{p} n (i_{\pi(1)}-j_{\pi(1)}),
		\end{equation}
		we have
		\begin{equation}\label{eq_sum_of_omega1}
			\begin{split}
				\sum_{n=0}^{p-1} R_{\mathbf{s}_{n}}(\tau)& = \sum_{i\in \Omega_{1}}\sum_{n=0}^{p-1} \omega_{q}^{f_{n,i}-f_{n,j}} + \sum_{i\in \Omega_{2}}\sum_{n=0}^{p-1} \omega_{q}^{f_{n,i}-f_{n,j}}\\
				&=p\sum_{i\in \Omega_{1}}\omega_{q}^{f_{0,i}-f_{0,j}},
			\end{split}
		\end{equation}
		where $\Omega_{1}=\{i\mid 0\leq i\leq p^m-1, i_{\pi(1)}=j_{\pi(1)}\},\Omega_{2}=\{i\mid 0\leq i\leq p^m-1, i_{\pi(1)}\neq j_{\pi(1)}\}$.
		
		Now consider the elements in $\Omega_{1}$. Since $i\neq j$, we can define $v$ to be the smallest integer for which $i_{\pi(v)}\neq j_{\pi(v)}$. Let $i^{(n)}$ and $j^{(n)}$ denote two integers whose $p$-ary representation satisfies
		\begin{equation}
			i^{(n)}_{\pi(k)}=\begin{cases}
				i_{\pi(k)}, & \text{if $k\neq v-1$}\\
				i_{\pi(k)}+n, & \text{if $k= v-1$}
			\end{cases},
			j^{(n)}_{\pi(k)}=\begin{cases}
				j_{\pi(k)}, & \text{if $k\neq v-1$}\\
				j_{\pi(k)}+n, & \text{if $k= v-1$}
			\end{cases}
		\end{equation}
		for $0\leq n\leq p-1$, then $i^{(n)}\in \Omega_{1}$. It means that we define an invertible map from the ordered pair $(i,j)$ to $(i^{(n)},j^{(n)})$ and both pairs contribute to (\ref{eq_sum_of_omega1}). 
		
		Note that
		\begin{equation}
			\begin{split}
				&f_{0,i^{(n)}}-f_{0,i}-f_{0,j^{(n)}} + f_{0,j} \\
				=& \frac{q}{p}\left(\sum_{k=1}^{m-1} a_{k} \left(i_{\pi(k)}^{(n)}  i_{\pi(k+1)}^{(n)} - i_{\pi(k)} i_{\pi(k+1)} - j_{\pi(k)}^{(n)}  j_{\pi(k+1)}^{(n)} + j_{\pi(k)} j_{\pi(k+1)}\right)\right)\\
				&+\frac{q}{p}\left(\sum_{k=h+1}^{m}v_{k}\left(i_{k}^{(n)}-i_{k} - j_{k}^{(n)} + j_{k}\right) \right) + d\sum_{k=1}^{h} p^{k-1} (i_{k}^{(n)}-i_{k} - j_{k}^{(n)} + j_{k}) \\
				=&\frac{q}{p} n a_{v-1} (i_{\pi(v)}-j_{\pi(v)}),
			\end{split}			
		\end{equation}
		we finally get
		\begin{equation}
			\sum_{n=0}^{p-1} \omega_{q}^{f_{0,i^{(n)}}-f_{0,j^{(n)}}} =\omega_{q}^{f_{0,i}-f_{0,j}} \sum_{n=0}^{p-1} \omega_{q}^{\frac{q}{p} n a_{v-1} (i_{\pi(v)}-j_{\pi(v)}) } =0.
		\end{equation}	
	
		Combining the above cases, we see
		\begin{equation}
			\sum_{n=0}^{p-1} R_{\mathbf{s}_{n}}(\tau) = 0.
		\end{equation}
		Therefore, $\{f_{0},\cdots,f_{p-1}\}$ forms a $q$-ary $(p,p^m)$-CS. This completes the proof.
	\end{proof}	}}	
	
	{\red{The above lemma provide a construction of $p^h$-ary CS extended from $p$-ary case, which means that the constructions of $\mathbf{\Phi}$ in the above sections can also be extended to $p^h$ case.}} Suppose that $\mathcal{M}=\{\mathbf{A}_{1},\cdots,\mathbf{A}_{L}\}$ is a set of matrices we have proposed {\red{in the above theorems}}. We define	
	\begin{equation}
		\mathbf{\Phi}_{\mathbf{A}_{k}}^{'} = [\mathbf{s}_{k}^{(0)},\mathbf{s}_{k}^{(1)},\cdots,\mathbf{s}_{k}^{(p^{m}-1)}].
	\end{equation}
	where $\mathbf{s}_{k}^{(c)}$ is the $p^h$-ary complex-valued sequence corresponding to the following EBF
	\begin{equation}\label{eq66}
		f_{k}^{(c)}(\mathbf{x})=\frac{q}{p} \mathbf{x}^{T}\mathbf{A}_{k}\mathbf{x} +\mathcal{L}_{c}(\mathbf{x}).
	\end{equation}
	{\red{It is easy to verify that $\mathbf{\Phi}_{\mathbf{A}_{k}}^{'}$ is orthogonal.}} Similar to Section \ref{section3}, the $p^h$-ary sequence set $\mathbf{\Phi}^{'}$ generated by $\mathcal{M}$ is defined as:
	\begin{equation}\label{eq67}
		\mathbf{\Phi}^{'}=[	\mathbf{\Phi}_{\mathbf{A}_{1}}^{'},\cdots,	\mathbf{\Phi}_{\mathbf{A}_{L}}^{'}].
	\end{equation}
	\red{Since, the quadratic part of (\ref{eq_q_ary_function}) and (\ref{eq66}) are similar to that of (\ref{eq_quadratic_functions}), using the results in Section \ref{sec4_1} and Section \ref{sec4_2}, we get the following Corollaries.

\begin{corollary}\label{cor1}
	Let $\mathcal{M}$ be the set of $\mathbf{A}_{k}$'s as described in Theorem \ref{th_coherence_Lp}. Then $\mathbf{\Phi}^{\prime}$, as described in (\ref{eq67}), is a $p^h$-ary sequence set with $p^{m+1}$ sequences, each of length $p^m$, ${\rm PAPR}(\mathbf{\Phi}^{\prime})\leq p$, and $\mu(\mathbf{\Phi}^{\prime})=\sqrt{\frac{1}{M}}$.
\end{corollary}
	\begin{proof}
			Since every sequence in $\mathbf{\Phi}^{\prime}$ is generated by a function with the form given in (\ref{eq_q_ary_function}), by Lemma \ref{lemma_q_ary_CS}, we obtain that each sequence in (\ref{eq_q_ary_function}) is from a $q$-ary $(p,p^m)$-CS. Then by Lemma \ref{lemma_upperboundPAPR_of_CSS}, we have ${\rm PAPR}(\mathbf{\Phi}^{\prime})=\max_{\mathbf{s}\in \mathbf{\Phi}^{\prime}} {\rm PAPR}(\mathbf{s}) \leq p$.
			
			The proof for the coherence can be derived directly from Theorem \ref{th_coherence_Lp}.
\end{proof}

\begin{corollary}\label{cor2}
Let $\mathcal{M}$ be the set of $\mathbf{A}_{k}$'s as described in Theorem \ref{construction2} or Theorem \ref{construction3}. Then $\mathbf{\Phi}^{\prime}$, as described in (\ref{eq67}), is a $p^h$-ary sequence set with $2\times p^{m+1}$ sequences, each of length $p^m$, ${\rm PAPR}(\mathbf{\Phi}^{\prime})\leq p$, and $\mu(\mathbf{\Phi}^{\prime})\leq\sqrt{\frac{1}{M}}$.
\end{corollary}

\begin{proof}
	The proof for PAPR is same as Corollary \ref{cor1}. The proof for the coherence can be derived directly from Theorem \ref{construction2} or Theorem \ref{construction3}, according to the choice of $\mathcal{M}$.
\end{proof}

\begin{corollary}\label{cor3}
Let $\mathcal{M}$ be the set of $\mathbf{A}_{k}$'s as described in Theorem \ref{construction_optimum_p=3_m_even_1} or Theorem \ref{construction_optimum_p=3_2}. Then $\mathbf{\Phi}^{\prime}$, as described in (\ref{eq67}), is a $3^h$-ary sequence set with $2\times 3^{m+1}$ sequences, each of length $3^m$, ${\rm PAPR}(\mathbf{\Phi}^{\prime})\leq 3$, and $\mu(\mathbf{\Phi}^{\prime})=\sqrt{\frac{1}{M}}$.
\end{corollary}

\begin{proof}
	The proof for PAPR is same as Corollary \ref{cor1}. The proof for the coherence can be derived directly from Theorem \ref{construction_optimum_p=3_m_even_1} or Theorem \ref{construction_optimum_p=3_2}, according to the choice of $\mathcal{M}$.
\end{proof}}

Finally, we compare the parameters of existing spreading sequences in table \ref{tab1}.

\begin{table}[h]
	\caption{The parameters of the spreading sequence matrices proposed till date. \label{tab1}}
	\begin{tabular}{|c|c|c|c|c|c|}
		\hline Reference & Constraint & \makecell[c]{Alphabet\\ size} & Coherence & \makecell[c]{Overloading \\ factor} & PAPR \\
		\hline 
		\multirow{2}{*}{\cite{Nam2020}} & $m$ is even & $2^h$ & $2 / \sqrt{M}$ & $3$ & $\leq 4$ \\
		\cline { 2 - 6 } & $m$ is odd & $2^h$  & $\sqrt{2/M}$ & $3$ & $\leq 2$ \\
		\hline 
		\multirow{2}{*}{\cite{Nam2021}} & $m=6,8,10$ & $2$ & $1 / \sqrt{M}$ & $5$ & $\leq 2$ \\
		\cline { 2 - 6 } & $m=5,7,9$& $2$ & $\sqrt{2/M}$ & $8$ & $\leq 2$ \\
		\hline 
		\multirow{2}{*}{\cite{LYB2022CL}} & $m$ is even& $2$ & $1 / \sqrt{M}$ & $3$ & $ \leq 2$ \\
		\cline { 2 - 6 } & $m \equiv 1(\bmod 4)$& $2$ & $\sqrt{2/M}$ & $4$ & $\leq 2$ \\
		\hline
		Theorem \ref{th_coherence_Lp}  & $m\geq 2$ & \red{$p$} & $1/\sqrt{M}$ & $p$ &$\leq p$\\ \hline
		Theorem \ref{construction2} & $m\geq 2$  & \red{$p$} & $\leq \sqrt{p/M}$ & $2p$ & $\leq p$ \\ \hline
		Theorem \ref{construction3} & $m\geq 2$ & \red{$p$} & $\leq \sqrt{p/M}$ & $2p$ & $\leq p$ \\ \hline
		Theorem \ref{construction_optimum_p=3_m_even_1}  & {\red{\makecell{$m$ is even, $m\geq 4$,\\$m\not\equiv 2 \bmod 3$}}} & \red{$3$} & $1/\sqrt{M}$ & $6$ & $\leq 3$\\ \hline
		Theorem \ref{construction_optimum_p=3_2}  & $m\geq 2$ & \red{$3$} & $1/\sqrt{M}$ & $6$ & $\leq 3$\\ \hline
		\red{Corollary \ref{cor1}}  & \red{\makecell{Similar to\\ Theorem \ref{th_coherence_Lp}}} & \red{$p^h$} & \red{$1/\sqrt{M}$} & \red{$p$} &\red{$\leq p$}\\ \hline
		\red{Corollary \ref{cor2}} & \red{\makecell{Similar to\\ Theorem \ref{construction2} or\\Theorem \ref{construction3}}}  & \red{$p^h$} & \red{$\leq \sqrt{p/M}$} & \red{$2p$} & \red{$\leq p$} \\ \hline
		\red{Corollary \ref{cor3}}  & {\red{\makecell{Similar to\\ Theorem \ref{construction_optimum_p=3_m_even_1} or\\Theorem \ref{construction_optimum_p=3_2}}}} & {\red{$3^h$}} & \red{$1/\sqrt{M}$} & \red{$6$} & \red{$\leq 3$}\\ \hline
	\end{tabular}
\end{table}

\section{Concluding Remarks}\label{sec5}
In this paper, we have proposed three infinite families of non-orthogonal spreading sequence matrices over $\mathbb{F}_p$, where $p$ is an odd prime, for uplink NOMA, using \red{EBFs}. We have represented the \red{EBFs} of degree two, in the form of matrices, and associated the coherence of the related spreading sequences with the rank of those matrices. To be more specific, the rank of the matrices are inversely proportional to the coherence of the spreading sequence matrices associated with the corresponding \red{EBF}. Each of the proposed spreading sequence matrices have low PAPR and low coherence. We have proved that we can achieve optimum coherence for $p=3$. For the other prime cases, we can achieve near-optimum coherence. We have also changed the alphabet size of the proposed sequence sets from $p$ or $3$ to $p^h$ or $3^h$, respectively. The proposed constructions hugely increases the overloading factor. For a given $p$, the overloading factor is $2p$. Specifically, for $p=3$, it is $6$. Due to computational complexity, we have only considered \red{EBFs} of degree two. As a future work, \red{EBFs} with higher degrees can be considered.

\bibliography{LKQ_REF}

\end{document}